\documentclass[twoside,11pt]{article}

\usepackage{blindtext}

\usepackage[normalem]{ulem}
\newcommand{\stkout}[1]{\ifmmode\text{\sout{\ensuremath{#1}}}\else\sout{#1}\fi}

\usepackage[abbrvbib, preprint]{jmlr2e}

\usepackage[american]{babel}
\usepackage[T1]{fontenc}
\usepackage[utf8]{inputenc}
\usepackage{mathtools} 
\usepackage{booktabs} 
\usepackage{tikz} 
\usetikzlibrary{positioning, quotes}
\usepackage{fancyhdr}
\usepackage[normalem]{ulem}
\usepackage{algorithm} 
\usepackage{algpseudocodex} 
\usepackage{cancel} 
\usepackage{subcaption}
\usepackage{multirow}
\usepackage{enumitem}
\usepackage{listings}

\DeclareMathOperator{\Perp}{\perp\!\!\!\perp}

\DeclareMathOperator{\pa}{\operatorname{pa}}
\DeclareMathOperator{\ch}{\operatorname{ch}}

\DeclareMathOperator{\nb}{\operatorname{nb}}
\DeclareMathOperator{\bd}{\operatorname{bd}}

\DeclareMathOperator{\dis}{\operatorname{dis}}
\DeclareMathOperator{\diff}{\operatorname{diff}}

\DeclareMathOperator{\mb}{\operatorname{mb}}

\DeclareMathOperator{\odds}{\operatorname{OR}}

\newcommand{\G}{{\mathcal G}}

\usepackage{lastpage}
\jmlrheading{23}{2025}{1-\pageref{LastPage}}{1/21; Revised 5/22}{9/22}{21-0000}{Trung Phung, Kyle Reese, Ilya Shpitser, Rohit Bhattacharya}

\ShortHeadings{Recursive Equations For Imputation Of MNAR Data With Sparse Pattern Support}{Phung, Reese, Shpitser, and Bhattacharya}
\firstpageno{1}

\author{\name Trung Phung \email tphung1@jhu.edu \\
       \addr Department of Computer Science\\
       Johns Hopkins University\\
       Baltimore, MD 21218, USA
       \AND
       \name Kyle Reese \email kreese15@jhu.edu \\
       \addr Department of Computer Science\\
       Johns Hopkins University\\
       Baltimore, MD 21218, USA
       \AND
       \name Ilya Shpitser \email ilyas@cs.jhu.edu \\
       \addr Department of Computer Science\\
       Johns Hopkins University\\
       Baltimore, MD 21218, USA
       \AND
       \name Rohit Bhattacharya \email rb17@williams.edu \\
       \addr Department of Computer Science\\
       Williams College\\
       Williamstown, MA 01267, USA}

\title{Recursive Equations For Imputation Of Missing Not At Random Data With Sparse Pattern Support}

\begin{document}
\maketitle

\begin{abstract}

A common approach for handling missing values in data analysis pipelines is multiple imputation via software packages such as  MICE \citep{van2011mice} and Amelia \citep{honaker2011amelia}. These packages typically assume the data are missing at random (MAR), and impose  parametric or smoothing assumptions upon the imputing distributions in a way that allows imputation to proceed even if not all missingness patterns have support in the data. 
Such assumptions  are unrealistic in practice, and induce model misspecification bias on any analysis performed after such imputation.

In this paper, we provide a principled alternative.  Specifically, we develop a new characterization for the full data law in graphical models of missing data. This characterization is constructive, is easily adapted for the calculation of imputation distributions for both MAR and MNAR (missing not at random) mechanisms, and is able to handle lack of support for certain patterns of missingness. We use this characterization to develop a new imputation algorithm---Multivariate Imputation via Supported Pattern Recursion (MISPR)---which uses Gibbs sampling, by analogy with the Multivariate Imputation with Chained Equations (MICE) algorithm, but which is consistent under both MAR and MNAR settings, and is able to handle missing data patterns with no support without imposing additional  assumptions beyond those already imposed by the missing data model itself.

In simulations, we show MISPR obtains comparable results to MICE when data are MAR, and superior, less biased results when data are MNAR.  Our characterization and imputation algorithm based on it are a step towards making principled missing data methods more practical in applied settings, where the data are likely both MNAR and sufficiently high dimensional to yield missing data patterns with no support at available sample sizes.

\end{abstract}

\begin{keywords}
  imputation, graphical models, pattern mixture factorization, Gibbs sampling
\end{keywords}

\section{Introduction}

Missing data is a pervasive problem in applied data analyses across a variety of disciplines in the empirical and social sciences \citep{pigott2001review, mcknight2007missing, eekhout2012missing, berchtold2019treatment}. Currently, imputation using software packages such as MICE \citep{van2011mice} and Amelia \citep{honaker2011amelia} is one of the most popular approaches for analyzing datasets that exhibit missingness \citep{sinharay2001use, azur2011multiple, sun2023deep}. However, these methods assume the data are missing at random (MAR), and impose  parametric or smoothing assumptions upon the imputing distributions in a way that allows imputation to proceed even if not all missingness patterns have support in the data. Such assumptions  are unrealistic in practice, and induce model misspecification bias on any analysis performed after such imputation.

A promising alternative approach to missing data modeling in complex multivariate settings is based on graphical models.  This approach provides a number of advantages for analysis, including a concise representation of the observed data likelihood, an intuitive causal interpretation of missingness mechanisms, as well as developed methods for identification theory \citep{robins1997non, daniel2012using, bhattacharya2019mid, mohan2021graphical, malinsky2021semiparametric, nabi2020mid, nabi2022causal, miao2024identification}.  Finally, graphical models can easily incorporate methods that address model uncertainty, either based on learning a graph using model selection procedures \citep{gain2018structure,tu2019causal}, or developing identification criteria that are sound in a potentially large union model of missingness processes \citep{mathur25common}.

Despite these advantages and a growing recognition of limitations of existing imputation based approaches \citep{sterne2009multiple, perkins2018principled, mathur2025pitfalls}, the use of graphical modeling methods for missing data in practical analyses has been limited.
This is because of additional methodological complexity of such methods: functionals identified under MNAR tend to have a complex form, and lead to semiparametric or weighting estimators that are challenging to implement.  Finally, the complexity of MNAR models of missingness makes it challenging to address missingness patterns with no support, a common complication in high dimensional applications with missing data.
All these complications mean that existing imputation algorithms
are far more likely to be incorporated into an analysis pipeline, due to their ease of implementation and use, and despite their known methodological disadvantages.




In this paper, we address these limitations of graphical models of missing data, and in doing so, provide a principled alternative to existing imputation methods.  Specifically, we develop a new characterization for the full data law in graphical models of missing data. This characterization is constructive, is easily adapted for the calculation of imputation distributions for both MAR and MNAR (missing not at random) mechanisms, and is able to handle lack of support for certain patterns of missingness. We use this characterization to develop a new imputation algorithm---Pattern mixture Imputaton with Gibbs Sampling (MISPR)---which uses Gibbs sampling, by analogy with the Multivariate Imputation with Chained Equations (MICE) algorithm, but which is consistent under both MAR and MNAR settings, and is able to handle missing data patterns with no support without imposing additional  assumptions beyond those already imposed by the missing data model itself. More specifically, we show how identification and imputation is possible by leveraging non-parametric restrictions encoded by a missing data graph when only the complete case pattern and small subsets of other missingness patterns are present. For completeness, we also provide an example of how to conduct these tasks without any complete case patterns.

In simulations, we show MISPR obtains comparable results to MICE when the data are MAR, and superior, less biased results when the data are MNAR.  Our characterization and imputation algorithm based on it are a step towards making principled missing data methods more practical in applied settings, where the data are likely both MNAR and sufficiently high dimensional to yield missing data patterns with no support at available sample sizes.

The remainder of this paper is organized as follows. In Section~\ref{sec:prelim}, we introduce notation, formalize our problem statement, and review the classes of missing data graphs considered in this paper. Section~\ref{sec:constructive_proof} presents PM-ID, our new constructive characterization of the full data law of these models and contrasts it with the identification approaches used in prior work (in particular, \cite{nabi2020mid} and \cite{malinsky2021semiparametric}) on graphical models that are not constructive. Section~\ref{sec:positivity-violations} extends our constructive characterization to settings with positivity violations, i.e., settings where certain patterns of missingness have no support. Section~\ref{sec:imputation_algorithm} describes our imputation algorithm based on this constructive characterization. Section~\ref{sec:experiments} provides numerical comparisons to MICE. Section~\ref{sec:comparison} provides detailed qualitative comparisons of our work to other work on MNAR missing data models, and finally, Section~\ref{sec:conclusion} concludes the paper.

\section{Problem Setup and Description Of Models}
\label{sec:prelim}

In this section, we formalize our problem statement, introduce relevant concepts related to graphical models of missing data, and provide basic intuition on how they are used.

\subsection{Missing Data Models}
\label{sec:prelim:missing_data}

Consider a missing data problem with $J$ fully observed random variables $O = \{O_1, \dots, O_J\}$, whose values are always observed in every data sample, and $K$ potentially missing random variables $X^{(1)} = \{X_1^{(1)}, \ldots, X_K^{(1)}\}$, whose values are missing in at least one of the sample. To describe the data generation process with missingness, we use two additional sets of observed random variables: the missingness indicators $R = \{R_1, \ldots, R_K\}$ and the observed proxies $X = \{X_1, \ldots, X_K\}$. The $X^{(1)}$ are considered hidden, while the observed data sample consists of values of $(O, X, R)$, determined via the missing data analogue of \emph{consistency} in causal inference: $X_k = X_k^{(1)}$ if $R_k=1$ and $X_k = ?$ if $R_k=0$. Each $O_j \in O$ takes values in $\mathcal{O}_j$, each $X^{(1)}_k \in X^{(1)}$ takes values in $\mathcal{X}^{(1)}_k$, and these state spaces are unrestricted. By definition, the state space of each missingness indicator $R_k \in R$ is $\{0, 1\}$, and the state space $\mathcal{X}_k$ of each observed proxy $X_k \in X$ is the state space of the corresponding potentially missing random variable augmented with the special value $``?"$ to indicate missingness, $\mathcal{X}_k = \mathcal{X}^{(1)}_k \cup \{?\}$.

A missing data model is a set of distributions defined over $X^{(1)} \cup O \cup R \cup X$ and satisfying missing data consistency.  That is, the joint distribution $p(X^{(1)}, O, R, X)$ may be written as
\begin{equation}
    \begin{aligned}
    p(X^{(1)}, O, R, X) = p(X^{(1)}, O, R) \times \prod_{k=1}^K p(X_k \mid R_k, X_k^{(1)}),
    \end{aligned}
    \label{eq:missing-data-fact}
\end{equation}
where each factor $p(X_k \mid R_k, X_k^{(1)})$ is deterministically defined by consistency. Following previous conventions \citep{bhattacharya2019mid, nabi2022causal}, the non-deterministic portion $p(X^{(1)}, O, R)$ is often referred to as the \emph{full law}, while the fully observed portion $p(O, X, R)$ is referred to as the \emph{observed law}. The full law can be further factorized as $p(X^{(1)}, O, R) = p(X^{(1)}, O) \times p(R \mid X^{(1)}, O)$, where $p(X^{(1)}, O)$ is known as the \emph{target law} and $p(R \mid X^{(1)}, O)$ is known as the \emph{missingness mechanism} or \emph{propensity score}.  Since $X$ are defined via fixed functions from $R$ and $X^{(1)}$, knowing the full law is equivalent to knowing the entire joint distribution $p(X^{(1)}, O, R, X)$.

For brevity, we sometimes use $[K] = \{1, \ldots, K\}$ as a shorthand for the set of indices of all missing variables. Further, if $Z \subseteq [K]$ is an index subset, we use the notation $R_{Z}, X^{(1)}_{Z}, X_Z$ to denote  the corresponding subset of variables, e.g., $R_{Z} = \{R_{k} \mid k \in Z\}$. 
Finally, for simplicity, we take statements such as $p(R=1 \mid X^{(1)}, O) > 0$ where no specific values are specified for some of the random variables to mean that the statement holds for all possible values of those random variables (e.g., $X^{(1)}, O$ in this case).

\paragraph{Goal} The goal of our work is to design identification and imputation algorithms for Missing Not At Random (MNAR) models---missing data models in which the missingness status of variables may depend on values that are themselves missing---that may also exhibit positivity violations. We pursue imputation as a goal, because it is often helpful as a pre-processing step in missing data problems as it allows downstream analyses to proceed as if there had been no missingness. Most often, parameters of interest are  functions of the target law $p(X^{(1)}, O)$, but occasionally they can be functions of the entire full law $p(X^{(1)}, O, R)$, for example in causal discovery applications \citep{gain2018structure, tu2019causal, liu2022greedy}. Note also that the target law is identified---i.e., can be expressed as a function of the observed law $p(O, X, R)$---if the full law is identified. Thus, in this work, we consider full law identification, and propose an imputation method for a wide class of MNAR models with full laws that factorize according to a graphical model.

\subsection{Graphical Models of Missing Data}
\label{sec:prelim:graphical_model}

\begin{figure}[t]
    \centering
    \begin{tikzpicture}[>=stealth, node distance=1.5cm]
    \def\d{1.5cm}
    \begin{scope}
        \path[->, very thick]
        node[] (x11) {$X^{(1)}_1$}
        node[right of=x11] (x21) {$X^{(1)}_2$}
        node[right of=x21] (x31) {$X^{(1)}_3$}
        
        node[below of=x11] (r1) {$R_1$}
        node[below of=x21] (r2) {$R_2$}
        node[below of=x31] (r3) {$R_3$}
        
        (r1) edge[blue] (r2)
        (r2) edge[blue] (r3)
        (x21) edge[blue] (r1)
        (x31) edge[blue] (r1)
        (x31) edge[blue] (r2)
        (x11) edge[blue] (r3)
        
        (x11) edge[blue] (x21)
        (x21) edge[blue] (x31)
        (x11) edge[blue, bend left] (x31)
        
        node[below of=r2, yshift=\d/3] () {(a)}
        ;
    \end{scope}

    \begin{scope}[xshift=5cm]
        \path[->, very thick]
        node[] (x11) {$X^{(1)}_1$}
        node[right of=x11] (x21) {$X^{(1)}_2$}
        node[right of=x21] (x31) {$X^{(1)}_3$}
        
        node[below of=x11] (r1) {$R_1$}
        node[below of=x21] (r2) {$R_2$}
        node[below of=x31] (r3) {$R_3$}
        
        (r1) edge[blue] (r2)
        (r2) edge[blue] (r3)
        (x31) edge[blue] (r1)
        (x31) edge[blue] (r2)
        (x11) edge[blue] (r3)
        (r1) edge[red, <->, bend right=25] (r2)
        
        (x11) edge[blue] (x21)
        (x21) edge[blue] (x31)
        (x11) edge[red, <->, bend left] (x31)
        
        node[below of=r2, yshift=\d/3] () {(b)}
        ;
    \end{scope}

    \begin{scope}[xshift=10cm]
        \path[->, very thick]
        node[] (x11) {$X^{(1)}_1$}
        node[right of=x11] (x21) {$X^{(1)}_2$}
        node[right of=x21] (x31) {$X^{(1)}_3$}
        
        node[below of=x11] (r1) {$R_1$}
        node[below of=x21] (r2) {$R_2$}
        node[below of=x31] (r3) {$R_3$}
        
        (r1) edge[-, brown] (r2)
        (r2) edge[-, brown] (r3)
        (r1) edge[-, brown, bend right] (r3)
        (x21) edge[blue] (r1)
        (x11) edge[blue] (r2)
        (x21) edge[blue] (r3)
        (x31) edge[blue] (r2)
        (x11) edge[blue] (r3)
        (x31) edge[blue] (r1)
        
        (x11) edge[blue] (x21)
        (x21) edge[blue] (x31)
        (x11) edge[blue, bend left] (x31)
        
        node[below of=r2, yshift=\d/3] () {(c)}
        ;
    \end{scope}
\end{tikzpicture}
\caption{Different kinds of missing data graphs: (a) A missing data DAG; (b) A missing data ADMG; (c) A missing data CG. Though each graph has different causal implications, all of them imply identification of the full data law as they encode the following set of independences for all $i \in \{1, 2, 3\}$: $X_i^{(1)} \Perp R_i \mid X_{-i}^{(1)}, R_{-i}$. }
\label{fig:m-graphs}
\end{figure}
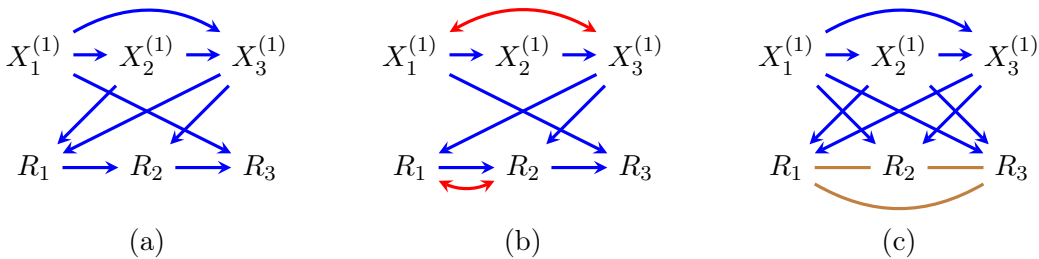

A graph $\G(V, E)$ consists of a set of vertices $V$ and edges $E$.  Similar to causal inference, much progress has been made in non-parametric missing data identification theory by representing missing random variables and their missingness indicators as vertices on a graph and using edges to encode both substantive mechanistic explanations of missingness and conditional independence relations between the missing variables and their missingness indicators. In particular, these conditional independences arise due to the absence of edges in the graph.

Thus far, three classes of graphs have been primarily used in the missing data literature: missing data DAGs (m-DAGs), missing data ADMGs (m-ADMGs), and missing data chain graphs (m-CGs), collectively referred to as \textit{m-graphs}. We describe these three classes here, and note that the identification and imputation methods described in following sections can be applied to datasets that come from distributions that factorize (or equivalently, satisfy the  Markov property) according to any of these three graph classes. 

\subsection*{Missing Data DAGs}

\renewcommand{\arraystretch}{1.5}

The most commonly used graphical models are those associated with directed acyclic graphs (DAGs). A DAG $\G(V, E)$ is a graph where the set of edges $E$ are restricted to be directed ($\rightarrow$) and no directed cycles (a series of edges that point in the same direction and have the same vertex as the start and end points, i.e., $V_i \rightarrow \cdots \rightarrow V_i$ ) are allowed.

In a missing data DAG, which we shorten to m-DAG following the convention in \cite{mohan2021graphical}, the vertices $V$ can be partitioned into the missing variables $X^{(1)}$, fully observed variables $O$, missingness indicators $R$, and observed proxies $X$. That is, for an m-DAG $\G(V, E)$, we have $V = X^{(1)} \cup O \cup R \cup X$. In addition to acyclicity, m-DAGs have two more restrictions, (i) Missingness indicators cannot have outgoing edges to variables in $X^{(1)}$ or $O$, and (ii) The only edges incident to each observed proxy $X_i$ are $X_i^{(1)} \rightarrow X_i \leftarrow R_i$, which is used to enforce the consistency assumption of missing data models.

Figure~\ref{fig:m-graphs}(a) shows an example of a valid m-DAG with three missing variables and no fully observed variables. For brevity, we  omit the portions of the graph that relate to the observed proxies, since they correspond to fully deterministic portions of the models and are not relevant for describing our identification and imputation algorithms.

The primary advantage of graphs is that they provide a concise  picture of independence assumptions. Let $\pa_\G(V_i)$ denote the parents of a vertex $V_i$ in $\G$, i.e., $\pa_\G(V_i) = \{ V_j \ |\  V_j \rightarrow V_i  \text{ in } \G  \}$. The statistical model of an m-DAG $\G(V, E)$ is the set of distributions $p(V)$ that factorize as $p(V) = \prod_{i=1}^{|V|} p(V_i \mid \pa_\G(V_i))$, with each conditional factor of an observed proxy $p(X_k \mid R_k, X_k^{(1)})$ appearing in this factorization being deterministic as before.

The factorization implies a set of independences that can be read off  an m-DAG using d-separation \citep{pearl1988probabilistic}. More directly relevant for our work, however, is a smaller set of independences that is similar to the set used to describe the local Markov property of DAG models \citep{pearl1988probabilistic, lauritzen1996graphical}. The local Markov property provides an equivalent (often simpler) description of these models. The independences are described using Markov blankets --- the Markov blanket of a variable $V_i$ is  the smallest set of variables in $V$ that makes $V_i$ conditionally independent of all other variables in the graph. In DAGs, the Markov blanket of a variable $V_i$ is known to be the parents of $V_i$, the children of $V_i$, and the co-parents of the children of $V_i$ \citep{pearl1988probabilistic}. That is, given a DAG $\G$, the Markov blanket can be described in terms of the sets defined in Table~\ref{table:mb-dags}.
\begin{table}[t]
\centering
\begin{tabular}{|l|l|l|}
\hline
\textbf{Set name} & \textbf{Notation} & \textbf{Definition} \\
\hline
Parents, Children & $\pa_\G(V_i), \ch_\G(V_i)$ & $\{V_j \ |\  V_j \rightarrow V_i  \text{ in } \G \}, \{V_j \ |\  V_i \rightarrow V_j  \text{ in } \G \}$ \\
Co-parents & $\text{co-pa}_\G(V_i)$ & $\{V_j \ |\  V_j \rightarrow V_k \leftarrow V_i  \text{ in } \G \}$ \\
Markov blanket & $\mb_\G(V_i)$ & $\pa_\G(V_i) \cup \ch_\G(V_i) \cup \text{co-pa}_\G(V_i)$\\
\hline
\end{tabular}
\caption{Definition of a Markov blanket in DAGs.}
\label{table:mb-dags}
\end{table}


This gives us the following set of local independences for every variable $V_i \in V$
\begin{align}
V_i \Perp V \setminus \mb_\G(V_i) \mid \mb_\G(V_i).
\label{eq:local-property}
\end{align}
As an example, if we apply the above property to the missing variables of the m-DAG model in Figure~\ref{fig:m-graphs}(a), we get the following independences:
\begin{align*}
&X_1^{(1)} \Perp R_1 \mid X_2^{(1)}, X_3^{(1)}, R_2, R_3 \\
&X_2^{(1)} \Perp R_2, R_3 \mid X_1^{(1)}, X_3^{(1)}, R_1 \\
&X_3^{(1)} \Perp R_3 \mid X_1^{(1)}, X_2^{(1)}, R_1, R_2
\end{align*}
In later sections we see how these independences help with identification and imputation.

We end this subsection with a note on causal interpretations of m-DAGs. Interpreted causally, a directed edge $V_i \rightarrow V_j$ in $\G$ indicates that $V_i$ is potentially a direct cause of $V_j$ relative to other variables on the graph \citep{spirtes2000causation, pearl2009causality}. This makes the MNAR nature of certain models fairly clear. For example, in Figure~\ref{fig:m-graphs}(a), the edge $X_1^{(1)} \rightarrow R_3$ clearly indicates that the missing values of one variable can cause the missingness of another. m-DAGs have appeared in numerous works, such as \cite{mohan2013missing, saadati2019adjustment, bhattacharya2019mid, nabi2020mid}, and \cite{chen2023causal}.

\subsection*{Missing Data ADMGs}
Missing data acyclic directed mixed graphs (m-ADMGs) are a useful extension of m-DAGs when some variables are not just missing but completely unobserved. Consider an m-DAG with vertices that can be partitioned into a set $V = X^{(1)} \cup O \cup R \cup X$ consisting of variables that are either fully observed or missing, and a set $H$ consisting of variables that are always unobserved or hidden. Given such an m-DAG $\G(V\cup H, E)$, an m-ADMG $\G(V, E')$  over the partially observed vertices is a mixed graph consisting of directed ($\rightarrow$) and bidirected ($\leftrightarrow$) edges obtained by applying the latent projection operator described in \cite{verma1990equivalence}. Latent projection of an m-DAG onto an m-ADMG maintains the acylicity property as well as the two additional restrictions of m-DAGs: No directed edges of the form $R_i \rightarrow X_j^{(1)}$  or $R_i \rightarrow O_j$ (though bidirected edges between these variables are permitted in m-ADMGs), and each observed proxy $X_k$ has only two parents $X_k^{(1)}$ and $R_k$.
From a causal perspective, a directed edge $V_i \rightarrow V_j$ in an m-ADMG maintains the same interpretation as before; a bidirected edge $V_i \leftrightarrow V_j$ can be interpreted (wlog) as the existence of one or more unobserved confounders $V_i \leftarrow H_k \rightarrow V_j$ between the two observed or partially observed variables \citep{evans2018margins}. An example of an m-ADMG is shown in Figure~\ref{fig:m-graphs}(b).

One of the main benefits of m-ADMGs is that they enable modeling of the (partially) observed margin $p(V)$ without explicitly considering parameterizations of latent variable models. The margin $p(V)$ satisfies a recursive factorization with respect to $\G$, known as the nested Markov factorization  \citep{Richardson.Evans.ea.2023.NestedMarkov}, that is significantly more complex than the DAG factorization. For the purpose of full law identification and imputation, however, it is again sufficient to focus on a small set of independences similar to an ordinary local Markov property for ADMGs. The extension of Markov blankets from DAGs to ADMGs is shown in Table~\ref{table:mb-admgs}.\footnote{The definition of boundary in Table~\ref{table:mb-admgs} is the same as the definition of an \emph{ordered} Markov blanket that was used to define an ordered local Markov property in \cite{richardson2003markov}.} Given an m-ADMG $\G(V, E)$, this modified definition of the Markov blanket yields a set of  independences with the same form as those in \eqref{eq:local-property} for each partially observed variable $V_i \in V$.

\begin{table}[t]
\centering
\begin{tabular}{|l|l|l|}
\hline
\textbf{Set name} & \textbf{Notation} & \textbf{Definition} \\
\hline
District & $\dis_\G(V_i)$ & $\{V_j \ |\  V_j \text{ has a bidirected path to } V_i \text{ in } \G \}  \cup \{V_i\}$ \\
Boundary & $\bd_\G(V_i)$ & $\pa_\G(V_i) \cup \dis_\G(V_i) \cup \bigg( \bigcup_{D_i \in \dis_\G(V_i)} \pa_\G(D_i) \bigg)$ \\
Markov blanket & $\mb_\G(V_i)$ & $\bd_\G(V_i) \cup \ch_\G(V_i) \cup \bigg( \bigcup_{C_i \in \ch_\G(V_i)} \bd_\G(C_i) \bigg) \setminus \{V_i\}$\\
\hline
\end{tabular}
\caption{Definition of a Markov blanket in ADMGs.}
\label{table:mb-admgs}
\end{table}


\subsection*{Missing Data CGs}

A chain graph (CG) $\G(V, E)$ is a mixed graph consisting of directed ($\rightarrow$) and undirected ($-$) edges such that there are no partially directed cycles (a series of edges from $V_i$ back to itself such that at least one directed edge is present and every edge is either undirected or a directed edge pointing away from $V_i$.) Chain graphs have been  useful in modeling interference and dependent data problems \citep{ogburn2020causal, pena2020unifying, bhattacharya2020causal, tchetgen2021auto}. From a causal perspective, directed edges in a CG maintain the same interpretation as before, while an undirected edge $V_i - V_j$ can be interpreted as a potential causal relationship involving feedback between the variables $V_i$ and $V_j$ \citep{lauritzen1996graphical}.

Missing data CGs (m-CGs) have become an important tool for analyzing the limits of non-parametric identification in MNAR graphical models. In particular, \cite{shpitser2016consistent} and \cite{malinsky2021semiparametric} describe an m-CG model\footnote{\cite{sadinle2016itemwise} independently described the same model without an m-CG representation.}, known as the \emph{no self-censoring model}, where each missing variable $X_i^{(1)}$ is permitted to cause the missingness of any other variable except itself, and all missingness indicators are pairwise connected via undirected edges. That is, the only contraints arising in the missingness mechanism of the model are due to the absence of self-censoring edges $X_i^{(1)} \rightarrow R_i$. Soon after, \cite{nabi2020mid} showed that all m-DAGs and m-ADMGs with $K$ missing variables and identified full data laws are in fact statistical sub-models of a no self-censoring m-CG defined on as many variables, i.e., they imply the same or more independence constraints as the corresponding m-CG. A three variable no self-censoring model is shown in Figure~\ref{fig:m-graphs}(c). Though they have very different causal implications, the result in \cite{nabi2020mid} implies that the m-DAG and m-ADMG in Figures~\ref{fig:m-graphs}(a, b) are statistical sub-models of the m-CG in Figure~\ref{fig:m-graphs}(c). In the following sections, we rely on this observation quite frequently to unify the imputation theory in our paper, so that it applies equally well to any missing data model that factorizes according to an m-DAG, m-ADMG, or m-CG with an identified full data law.

We end by noting relevant statistical properties of CGs. The statistical model of a CG $\G(V, E)$ is the set of distributions $p(V)$ that satisfy a two-level factorization involving undirected connected components of $\G$ \citep{lauritzen1996graphical}.\footnote{A more general version of the factorization is provided in \cite{shpitser2023lauritzen}.} Once again, in our paper, we rely only on a local Markov property of m-CGs defined using an appropriately modified definition of the Markov blanket shown in Table~\ref{table:mb-cgs}. Given an m-CG $\G(V, E)$, applying this definition provides a list of constraints for each $V_i \in V$ of the form shown in~\eqref{eq:local-property}.
\begin{table}[t]
\centering
\begin{tabular}{|l|l|l|}
\hline
\textbf{Set name} & \textbf{Notation} & \textbf{Definition} \\
\hline
Neighbors & $\nb_\G(V_i)$ & $\{V_j \ |\  V_j - V_i \text{ in } \G \}$ \\
Co-parents & $\text{co-pa}_\G(V_i)$ & $\{V_j \mid V_j \rightarrow \circ \leftarrow V_i \text{ in } \G \text{ where } \circ$ \text{is one or more } \\
& & $\text{ vertices connected by an undirected path} \}$ \\
Markov blanket & $\mb_\G(V_i)$ & $\pa_\G(V_i) \cup \nb_\G(V_i) \cup \ch_\G(V_i) \cup  \text{co-pa}_\G(V_i)$\\
\hline
\end{tabular}
\caption{Definition of a Markov blanket in CGs.}
\label{table:mb-cgs}
\end{table}

\section{New Constructive Identification Method For Strictly Positive Full Laws}
\label{sec:constructive_proof}

In this section, we present a new identification method for the full law of missing data models based on m-graphs. To motivate the need for this new method, we briefly summarize the identification strategy and proof used in prior work by \cite{nabi2020mid} and \cite{malinsky2021semiparametric}. While this previous method is sound and complete for identification of the full law in m-DAGs and m-ADMGs\footnote{It has not been explicitly shown to be complete for m-CGs, but rather it was shown by \cite{nabi2020mid} that all m-DAGs and m-ADMGs with identified full laws are sub-models of a no self-censoring m-CG model studied in \cite{shpitser2016consistent}, \cite{sadinle2016itemwise}, and \cite{malinsky2021semiparametric}.}, we show that it is not constructive, and thus difficult to use in practice. In contrast, we show that our method is constructive, and it remains sound and complete, allowing for the design of an imputation algorithm in Section~\ref{sec:imputation_algorithm}. Further, this constructive argument is easily extended to scenarios with positivity violations, as we show in Section~\ref{sec:positivity-violations}. 

The key difference between the two methods lies in the factorizations used. \cite{nabi2020mid} and \cite{malinsky2021semiparametric} use an odds ratio factorization \citep{chen2007semiparametric} of the missingness mechanism, while we use a pattern mixture factorization \citep{little1993pattern} of the entire full law. Henceforth, for brevity, we will refer to the former as the OR-ID (odds ratio identification) method and the latter as the PM-ID (pattern mixture identification) method. As will be demonstrated, PM-ID allows both identification and imputation to be formulated in terms of extrapolation densities that can be derived explicitly in terms of the observed law.

Both OR-ID and the new PM-ID algorithm we present in this section require strict positivity of the full law. We relax this assumption in Section~\ref{sec:positivity-violations}, where we discuss certain common positivity violations in missing data problems.

\subsection{Graphical Missing Data Models with Identified Full Laws}
\label{sec:constructive_proof:model}

As mentioned earlier, we focus on missing data models with identified full laws, and in this section, will focus on ones that have strictly positive full laws. More formally,
consider a missing data model $p(X^{(1)}, O, R, X)$ that factorizes according to an m-graph $\G(V, E)$, where $V=X^{(1)} \cup O \cup R \cup X$. We make the following positivity assumptions that have been standard in prior work on missing data graphs. For some small constant $\epsilon > 0$, we have
\begin{enumerate}
    \setlength\itemindent{1em}
    \item[\textbf{(P1)}] Positivity of the missingness mechanism: $p(R \mid  X^{(1)}, O) \geq \epsilon$.
    \item[\textbf{(P2)}] Positivity of the target law: $p(X^{(1)}, O) \geq \epsilon$.
\end{enumerate}
Note by our convention, these conditions hold for all values of $R, X^{(1)}, O$, and (P1, P2) together imply strict positivity of the full law $p(X^{(1)}, O, R)$.

From prior work in \cite{bhattacharya2019mid}, \cite{nabi2020mid}, and \cite{malinsky2021semiparametric}, we know that under positivity assumptions (P1, P2), an m-graph $\G$ implies identification of the full law $p(X^{(1)}, O, R)$ if the following structural assumption\footnote{(S1) here is a succinct and equivalent description of previous graphical criteria posed in \cite{nabi2020mid} phrased in terms of absence of self-censoring edges and \emph{colluders} (colliders of the form $X_i^{(1)} \rightarrow R_j \leftarrow R_i$) for m-DAGs, and absence of self-censoring edges and \emph{colluding paths} (paths between $X_i^{(1)}$ and $R_i$ where every intermediate vertex is a collider $\rightarrow \circ \leftarrow$, $\leftrightarrow \circ \leftarrow$, or $\rightarrow \circ \leftrightarrow$) for m-ADMGs.} is satisfied,
\begin{enumerate}
    \setlength\itemindent{1em}
    \setlength\itemsep{0em}
    \item[\textbf{(S1)}] No missingness indicator lies in the Markov blanket of the variable it indicates.
    \item[] That is, for all $i \in [K]$, we have $R_i \notin \mb_{\G}(X_i^{(1)})$,
\end{enumerate}
where the definition of $\mb_\G(R_i)$ is inferred from the appropriate Tables~\ref{table:mb-dags},~\ref{table:mb-admgs}, or~\ref{table:mb-cgs} depending on whether $\G$ is an m-DAG, m-ADMG, or m-CG respectively.

For m-DAG and m-ADMG models satisfying strict positivity, \cite{nabi2020mid} further show that the graphical criterion (S1) is sound and complete for identification of the full law: If $\G$ satisfies (S1), then the full law is provably identified; if $\G$ does not satisfy (S1), then the full law is provably \emph{not} identified. 

It is easy to confirm that all three m-graphs in Figure~\ref{fig:m-graphs} satisfy (S1), and thus imply identification of the full law. Particularly, the following set of conditional independences can be derived as a direct consequence of (S1) and the local Markov properties of m-graphs described in the previous section: Since each missingness indicator does not appear in the Markov blanket of the corresponding variable it indicates,
\begin{align}
    \text{for all } i \in [K], \text{ we have } X_i^{(1)}  \Perp  R_i \mid X^{(1)}_{-i}, R_{-i}, O.
    \label{eq:ri_indep_xi}
\end{align}
These independences play a key role in non-parametric identification of the full law in all three models. We  demonstrate this first through a non-constructive identification method (OR-ID) from prior work, and then with a novel constructive identification method (PM-ID).

\subsection{Non-constructive method: OR-ID \citep{nabi2020mid, malinsky2021semiparametric}}
\label{sec:constructive_proof:or-id}

OR-ID uses an odds ratio factorization \citep{chen2007semiparametric} of $p(R\mid X^{(1)}, O)$ to prove identifiability of the full law whenever (P1, P2, S1) hold.
First, by Bayes rule the full law factorizes as,
\begin{align*}
p(X^{(1)}, O, R) = \frac{p(X^{(1)}, O, R=1)}{p(R=1 \mid O, X^{(1)})} \cdot p(R \mid X^{(1)}, O).
\end{align*}
The numerator in the above fraction is a function of the observed law by consistency, $p(X^{(1)}, O, R=1) = p(X, O, R=1)$. Thus, identification of the full law can be reduced to identification of the missingness mechanism $p(R\mid X^{(1)}, O)$ for all possible values of $R, X^{(1)}, O$. 
For any arbitrary choice of reference values for random variables $A$ and $B$, the odds ratio factorization of a joint conditional distribution $p(A, B \mid C)$ is \citep{chen2007semiparametric},
\begin{align*}
p(A, B \mid C) &= \frac{1}{Z(C)} \cdot p(A \mid B=b_0, C) \cdot p(B \mid A=a_0, C) \cdot \odds(A, B \mid C),
\end{align*}
where the odds ratio function $\odds(A, B \mid C)$ is defined as,
\begin{align*}
\odds(A, B \mid C) &= \frac{p(A\mid B, C)}{p(A=a_0 \mid B, C)} \cdot \frac{p(A=a_0 \mid B=b_0, C)}{p(A\mid B=b_0, C)}, 
\end{align*}
and $Z(C) = \sum_{A, B} p(A \mid B=b_0, C) \cdot p(B \mid A=a_0, C) \cdot \odds(A, B \mid C)$ is a normalizing term.

When the odds ratio factorization is applied recursively to the missingness mechanism with reference values of $1$ for each $R_k$, it gives us the following factorization of $p(R\mid X^{(1)}, O)$,
{\small
\begin{align*}
    p(R \mid X^{(1)}, O) = \frac{1}{Z(X^{(1)}, O)} \cdot \prod_{k=1}^{K} p(R_k \mid R_{-k}=1, X^{(1)}, O) \cdot \prod_{k=2}^K \odds(R_k, R_{\prec k} \mid R_{\succ k}=1, X^{(1)}, O), \label{eq:or_factorization}
\end{align*}}
where $Z(X^{(1)}, O)$ is a normalizing term as before and the sets $R_{-k}, R_{\prec k}, R_{\succ k}$ are defined as $R \setminus \{R_k\},  \{R_1, \dots, R_{k-1}\}$, and $\{R_{k+1}, \dots, R_K\}$ respectively.

\cite{nabi2020mid} and \cite{malinsky2021semiparametric} provide an argument for why each of the terms in this factorization are functions of observed data. Rather than describing the identification argument in full generality, we show a specific example of it applied to Figures~\ref{fig:m-graphs}(a, b, c), and demonstrate why it is not constructive. This makes  OR-ID difficult to use in practical settings, as directly fitting the observed data likelihood of missing data models can be computationally intractable even for small problems (due to the normalizing function and the need to marginalize over the missing variables in each row of data). \cite{malinsky2021semiparametric} provide semiparametric estimators of the missingness mechanism, but these estimators are also not scalable and require fairly involved estimation subroutines that are not easily tailored for different missing data graphs. Moreover, the non-parametric influence function may not be efficient in the working m-graph model in cases where it is a submodel of the corresponding no self-censoring model studied by them. 

\subsubsection{An example of identification using OR-ID}
\label{sec:constructive_proof:or-id:example}

\begin{example}

We apply OR-ID to Figures~\ref{fig:m-graphs}(a, b, c) as an example of how it works and why it does not provide an explicit functional for estimation. The arguments below apply to all three m-graphs. The odds ratio factorization of the missingness mechanism $p(R\mid X^{(1)})$ is,
\begin{equation*}
    \small
    \begin{aligned}
        p(R \mid X^{(1)}) &= \frac{1}{Z(X^{(1)})} \cdot
        p(R_1 \mid R_{23}=1, X^{(1)})
        p(R_2 \mid R_{13}=1, X^{(1)})
        p(R_3 \mid R_{12}=1, X^{(1)})
        \\
        & \times \operatorname{OR}(R_2, R_1 \mid R_3=1, X^{(1)})
        \operatorname{OR}(R_{12}, R_3 \mid X^{(1)}).
    \end{aligned}
\end{equation*}
Further decomposition of the term $\operatorname{OR}(R_{12}, R_3 \mid X^{(1)})$ gives
\begin{equation*}
    \small
    \begin{aligned}
        p(R \mid X^{(1)}) &= \frac{1}{Z(X^{(1)})} \cdot
        p(R_1 \mid R_{23}=1, X^{(1)})
        p(R_2 \mid R_{13}=1, X^{(1)})
        p(R_3 \mid R_{12}=1, X^{(1)})
        \\
        & \times
        \operatorname{OR}(R_{2}, R_1 \mid R_3=1, X^{(1)})
        \operatorname{OR}(R_{2}, R_3 \mid R_1=1, X^{(1)})
        \operatorname{OR}(R_{1}, R_3 \mid R_2=1, X^{(1)})
        \\
        & \times
        \frac{\operatorname{OR}(R_{1}, R_2 \mid R_3, X^{(1)})}{\operatorname{OR}(R_{1}, R_2 \mid R_3=1, X^{(1)})}.
    \end{aligned}
\end{equation*}

First, we identify the univariate conditional factors. We have,
\begin{align*}
p(R_1 \mid R_{23}=1, X^{(1)}) = p(R_1 \mid R_{23}=1, X_{23}^{(1)}) = p(R_1 \mid R_{23}=1, X_{23}),
\end{align*}
where the first equality follows from the independence $R_1^{(1)} \Perp X_1 \mid X^{(1)}_{23}, R_{23}$ appearing in \eqref{eq:ri_indep_xi}, and the second equality follows from consistency. A similar argument yields the identification of $p(R_2 \mid R_{13}=1, X^{(1)})$ and $p(R_3 \mid R_{12}=1, X^{(1)})$.

Second, to identify the odds ratios, we use their symmetry property and write them in two ways. For example,  consider $\operatorname{OR}(R_{1}, R_2 \mid R_3=1, X^{(1)})$ written in two ways,
\begin{equation*}
    \small
    \begin{aligned}
        \operatorname{OR}(R_{1}, R_2 \mid R_3=1, X^{(1)})
        &=
        \frac{p(R_1 \mid R_{2}, R_3=1, X^{(1)})}{p(R_1=1 \mid R_{2}, R_3=1, X^{(1)})} \cdot
        \frac{p(R_1=1 \mid R_{2}=1, R_3=1, X^{(1)})}{p(R_1 \mid R_{2}=1, R_3=1, X^{(1)})} \\
        &=\frac{p(R_2 \mid R_{1}, R_3=1, X^{(1)})}{p(R_2=1 \mid R_{1}, R_3=1, X^{(1)})} \cdot
        \frac{p(R_2=1 \mid R_{1}=1, R_3=1, X^{(1)})}{p(R_2 \mid R_{1}=1, R_3=1, X^{(1)})}.
    \end{aligned}
\end{equation*}
From the first equality we can infer that this odds ratio is not a function of $X_1^{(1)}$, as $R_1 \Perp X_1^{(1)} \mid X^{(1)}_{23}, R_{23}$ due to \eqref{eq:ri_indep_xi}. From the second equality, we can infer that the odds ratio is not a function of $X_2^{(1)}$, as $R_2 \Perp X_2^{(1)} \mid X^{(1)}_{23}, R_{23}$. Together with the previous claim, this implies that the odds ratio depends only on the joint distribution $p(R_{1}, R_2, X_3^{(1)}, R_3=1)$, which is identified by consistency. Similar arguments lead to the identification of the other two odds ratios $\operatorname{OR}(R_{2}, R_3 \mid R_1=1, X^{(1)})$ and $\operatorname{OR}(R_{1}, R_3 \mid R_2=1, X^{(1)})$.

Third, the 3-way interaction term can be written in three ways, due to symmetry,
\begin{equation*}
    \small
    \begin{aligned}
        \frac{\operatorname{OR}(R_{1}, R_2 \mid R_3, X^{(1)})}{\operatorname{OR}(R_{1}, R_2 \mid R_3=1, X^{(1)})}
        =
        \frac{\operatorname{OR}(R_{2}, R_3 \mid R_1, X^{(1)})}{\operatorname{OR}(R_{2}, R_3 \mid R_1=1, X^{(1)})}
        =
        \frac{\operatorname{OR}(R_{1}, R_3 \mid R_2, X^{(1)})}{\operatorname{OR}(R_{1}, R_3 \mid R_2=1, X^{(1)})}.
    \end{aligned}
\end{equation*}
By the independences in \eqref{eq:ri_indep_xi}, the term on the left does not depend on $X^{(1)}_{12}$, the middle term does not depend on $X^{(1)}_{23}$, and the right one does not depend on $X^{(1)}_{13}$. Thus, the interaction term does not depend on any of the missing variables $X^{(1)}_{123}$, and is identified.

Finally, identification of all these terms implies identification of the normalization function $Z(X^{(1)})$, and thus the propensity score $p(R \mid X^{(1)})$ and the full law $p(X^{(1)}, R)$.

\end{example}

\subsubsection{OR-ID is non-constructive}

Although OR-ID establishes identification for all terms in the factorization, it does not provide explicit expressions for the odds ratios and the higher order interaction terms in terms of the observed law. For example, despite the fact that OR-ID establishes that $\operatorname{OR}(R_{1}, R_2 \mid R_3=1, X^{(1)})$ is a function of the observed joint distribution $p(R_{1}, R_{2}, R_{3}=1, X^{(1)}_3) = p(R_{1}, R_{2}, R_{3}=1, X_3)$, we have the following inequality:
\begin{equation*}
    \small
    \begin{aligned}
        \operatorname{OR}(R_{1}, R_2 \mid R_3=1, X^{(1)})
        & \neq
        \frac{p(R_1 \mid R_{2}, R_{3}=1, X^{(1)}_3)}{p(R_1=1 \mid R_{2}, R_{3}=1, X^{(1)}_3)} \cdot
        \frac{p(R_1=1 \mid R_{2}=1, R_{3}=1, X^{(1)}_3)}{p(R_1 \mid R_{2}=1, R_{3}=1, X^{(1)}_3)},
    \end{aligned}
\end{equation*}
In other words, OR-ID establishes that this odds ratio can be expressed as some function of observed data $f(p(R_1, R_2, R_3=1, X_3))$, however, it does not tell us exactly what this function $f$ is. Further, the above inequality demonstrates this function $f$ could be quite complicated in general, as the odds ratios often do not ``collapse'' into simpler atomic expressions involving only observed data; see \cite{didelez2010graphical} and \cite{didelez2022logic} for some special cases where odds ratios are collapsible in the presence of selection bias. The same non-collapsibility issue is true for the other odds ratio and interaction terms appearing in the OR factorization. This motivates our new constructive identification method.

\subsection{PM-ID: A constructive identification method}
\label{sec:constructive_proof:pm-id}

PM-ID is an identification method that uses a pattern mixture factorization \citep{little1993pattern} of the full law. In what follows, we call an ordered tuple $r=(r_1, \dots, r_K$) a \textit{missingness pattern}, and use $\mathcal{R} = \{0,1\}^K$ to denote the set of all possible missingness patterns. For each missingness pattern $r \in \mathcal{R}$, the set $\mathbb{M}(r) := \{i \in [K] \mid  r_i = 0\}$ contains the indices of all the missing variables that are unobserved in the pattern $r$, while the set $\mathbb{O}(r) = [K] \setminus \mathbb{M}(r)$ contains the indices of all the missing variables that are observed in the pattern. 



For any missingness pattern $r \in \mathcal{R}$, the pattern mixture factorization (PM factorization) of the full law at $r$ is given by, 
\begin{equation}
    \label{eq:pm_factorization}
    \small
    \begin{aligned}
    p(X^{(1)}, O, R=r)
    &=
    p(X^{(1)}_{\mathbb{M}(r)} \mid X_{\mathbb{O}(r)}^{(1)}, O, R=r)
    \cdot p(X_{\mathbb{O}(r)}^{(1)} \mid O, R=r) \cdot p(O, R=r)
    &(\text{chain rule})
    \\
    &=
    \underbrace{
        p(X^{(1)}_{\mathbb{M}(r)} \mid X_{\mathbb{O}(r)}, O, R=r)
    }_{\text{extrapolation density}}
    \cdot
    \underbrace{
        p(X_{\mathbb{O}(r)} \mid O, R=r)
    }_{\text{interpolation density}}
    \cdot p(O, R=r).
    &(\text{consistency})
    \end{aligned}
\end{equation}
The PM factorization is well-defined if $p(X_{\mathbb{O}(r)}, O, R=r)$ is positive, which is guaranteed by positivity assumptions (P1, P2).

The term $p(X_{\mathbb{O}(r)} \mid O, R=r)$ in \eqref{eq:pm_factorization} is known as the \textit{interpolation density}  for pattern $r$, while the term $p(X^{(1)}_{\mathbb{M}(r)} \mid X_{\mathbb{O}(r)}^{(1)}, O, R=r)$ is known as the \emph{extrapolation density}. Since the interpolation density always consists of variables that are observed in the given pattern, it is always identified by consistency. Identifiability of the extrapolation density, however, must be established based on model assumptions, and the full joint distribution $p(X^{(1)}, O, R=r)$ for a certain pattern $r$ is identified if and only if the corresponding extrapolation density is identified. Thus, another way to establish identifiability of the full law $p(X^{(1)}, O, R)$ that is distinct from the strategy used by OR-ID, is to argue that the extrapolation density $p(X^{(1)}_{\mathbb{M}(r)} \mid X_{\mathbb{O}(r)}^{(1)}, O, R=r)$ is identified for all possible missingness patterns $r \in {\cal R}$.

This is the premise for the PM-ID algorithm proposed here. We also establish in Section~\ref{sec:imputation_algorithm} that knowing the extrapolation density for a certain missingness pattern $r$ is sufficient for imputation of data rows matching the pattern $r$. Hence, the extrapolation density plays a key role in our identification and imputation strategies. The following example makes the intuition behind using the PM factorization for identification more concrete, and leads to a general algorithm in the next subsection.

\subsubsection{An example of constructive identification with PM factorization}
\label{sec:constructive_proof:pm-id:example}

\begin{example}

Consider the m-DAG in Figures~\ref{fig:m-graphs}(a, b, c). We   compute explicit identifying functionals of $p(X^{(1)}, R=r)$ for every pattern $r \in \{0,1\}^3$ according to a partial order---starting with the complete case pattern $R=111$, followed by any patterns where only one variable is unobserved (e.g., $R=011$),  then patterns where only two variables are unobserved (e.g, $R=001$), and finally the pattern where no variable is observed $R=000$. 
\begin{enumerate}[leftmargin=*]
    \item The complete case pattern $R=111$:
        By consistency, $p(X^{(1)}, R=111) = p(X, R=111)$.
    \item Patterns with only one unobserved variable: Say we first process the pattern $R=011$ (the ordering between these patterns is not important). The PM factorization gives us,
        \begin{equation*}
            \begin{aligned}
                p(X^{(1)}, R=011) &= p(X^{(1)}_1 \mid X^{(1)}_{23}, R=011) \cdot p(X_{23} \mid R=011) \cdot p(R=011).
            \end{aligned}
        \end{equation*}
        The extrapolation density is identified as,
        \begin{equation*}
            \begin{aligned}
                p(X^{(1)}_1 \mid X^{(1)}_{23}, R=011) &= p(X^{(1)}_1 \mid X_{23}^{(1)}, R=111) = p(X_1 \mid X_{23}, R=111),
            \end{aligned}
        \end{equation*}
    where the first equality follows from $X_1^{(1)} \Perp R_1 \mid X^{(1)}_{23}, R_{23}$ and the second follows from consistency. The identification arguments for $R=101$ and $R=110$ are similar, using the independences $X_2^{(1)} \Perp R_2 \mid X^{(1)}_{13}, R_{13}$ and $X_3^{(1)} \Perp R_3 \mid X^{(1)}_{12}, R_{12}$ respectively. For example, for the pattern $R=101$ we have,
    \begin{equation*}
        \begin{aligned}
            p(X^{(1)}, R=101)
            &= p(X_2^{(1)} \mid X_{13}^{(1)}, R=101) \cdot p(X_{13} \mid R=101) \cdot p(R=101), \\
            &= p(X_2 \mid X_{13}, R=111) \cdot p(X_{13} \mid R=101) \cdot p(R=101).
        \end{aligned}
    \end{equation*}
        
    That is, identification of the joint distribution of patterns with one unobserved variable follows by re-expressing them as functions of the complete case pattern. In the next step of PM-ID, we express the patterns with two unobserved variables as functions of the (now identified patterns) with only one unobserved variable.
        
    \item Patterns with two unobserved variables: Consider the pattern $R=001$ and its corresponding PM factorization,
        \begin{equation*}
            \begin{aligned}
                p(X^{(1)}, R=001) &= p(X^{(1)}_{12} \mid X^{(1)}_{3}, R=001) \cdot p(X_{3} \mid R=001) \cdot p(R=001).
            \end{aligned}
        \end{equation*}
        A key to identifying such patterns is that the extrapolation density is identified if  $p(X^{(1)}_{1} \mid X^{(1)}_{23}, R=001)$ and $p(X^{(1)}_{2} \mid X^{(1)}_{13}, R=001)$ are identified, as these conditional distributions are sufficient for Gibbs sampling from the joint distribution \citep{geman1984stochastic}. Henceforth, we refer to such conditionals as Gibbs factors. We identify the Gibbs factors for this extrapolation density  using different conditional independences as follows,
        \begin{equation*}
            \begin{aligned}
                p(X^{(1)}_{1} \mid X^{(1)}_{23}, R=001)
                &= p(X^{(1)}_{1} \mid X^{(1)}_{23}, R=101),
                &(X_1^{(1)} \Perp R_1 \mid X^{(1)}_{23}, R_{23})
                \\
                p(X^{(1)}_{2} \mid X^{(1)}_{13}, R=001)
                &= p(X^{(1)}_{2} \mid X^{(1)}_{13}, R=011).
                &(X^{(1)}_2 \Perp R_2 \mid X^{(1)}_{13}, R_{13})
            \end{aligned}
        \end{equation*}
        The conditional distributions on the right hand side of these equations are identified, as they are functions of the joint distributions $p(X^{(1)}, R=101)$ and $p(X^{(1)}, R=011)$ that were identified in the previous step.
        
        We can repeat this argument for pattern $R=010$:
        The extrapolation density $p(X^{(1)}_{13} \mid X^{(1)}_{2}, R=010)$ also consists of two Gibbs factors identified as follows,
        \begin{equation*}
            \begin{aligned}
                p(X^{(1)}_{1} \mid X^{(1)}_{23}, R=010)
                &= p(X^{(1)}_{1} \mid X^{(1)}_{23}, R=110),
                &(X_1^{(1)} \Perp R_1 \mid X^{(1)}_{23}, R_{23})
                \\
                p(X^{(1)}_{3} \mid X^{(1)}_{12}, R=010)
                &= p(X^{(1)}_{3} \mid X^{(1)}_{12}, R=011),
                &(X^{(1)}_3 \Perp R_3 \mid X^{(1)}_{12}, R_{12})
            \end{aligned}
        \end{equation*}
        where the conditionals on the right hand side are again functions of joint distributions involving only one unobserved variable that were identified in the previous step.
    
        Finally for pattern $R=100$, the Gibbs factors for the extrapolation density $p(X^{(1)}_{23} \mid X^{(1)}_{1}, R=100)$ can also be identified as,
        \begin{equation*}
            \begin{aligned}
                p(X^{(1)}_{2} \mid X^{(1)}_{13}, R=100)
                &= p(X^{(1)}_{2} \mid X^{(1)}_{13}, R=110),
                &(X^{(1)}_2 \Perp R_2 \mid X^{(1)}_{13}, R_{13})
                \\
                p(X^{(1)}_{3} \mid X^{(1)}_{12}, R=100)
                &= p(X^{(1)}_{3} \mid X^{(1)}_{12}, R=101).
                &(X^{(1)}_3 \Perp R_3 \mid X^{(1)}_{12}, R_{12})
            \end{aligned}
        \end{equation*}
    \item For the final pattern $R=000$,
        the extrapolation density is $p(X^{(1)} \mid R=000)$. This is identified because all three Gibbs factors can be expressed as functions of of $p(X^{(1)}, R=100)$, $p(X^{(1)}, R=010)$ and $p(X^{(1)}, R=001)$ that were identified in the previous step.
        \begin{equation*}
            \begin{aligned}
                p(X^{(1)}_{1} \mid X^{(1)}_{23}, R=000)
                &= p(X^{(1)}_{1} \mid X^{(1)}_{23}, R=100),
                &(X_1^{(1)} \Perp R_1 \mid X^{(1)}_{23}, R_{23})
                \\
                p(X^{(1)}_{2} \mid X^{(1)}_{13}, R=000)
                &= p(X^{(1)}_{2} \mid X^{(1)}_{13}, R=010),
                &(X^{(1)}_2 \Perp R_2 \mid X^{(1)}_{13}, R_{13})
                \\
                p(X^{(1)}_{3} \mid X^{(1)}_{12}, R=000)
                &= p(X^{(1)}_{3} \mid X^{(1)}_{12}, R=001).
                &(X^{(1)}_3 \Perp R_3 \mid X^{(1)}_{12}, R_{12})
            \end{aligned}
        \end{equation*}
\end{enumerate}
Since $p(X^{(1)}, R=r)$ is identified for every $r \in \mathcal{R}$, the full law is identified.

\end{example}

\subsubsection{The PM-ID Algorithm}
\label{sec:constructive_proof:theorem}

The example in the previous subsection provides the following insights that we use to phrase a general identification algorithm based on the pattern mixture factorization:
\begin{enumerate}
    \item The distribution $p(X^{(1)}, O, R=r)$ for a specific pattern $r \in {\cal R}$ is identified if and only if the corresponding extrapolation density $p(X^{(1)}_{\mathbb{M}(r)} \mid X_{\mathbb{O}(r)}^{(1)}, O, R=r)$ is identified.
    \item To identify an extrapolation density $p(X^{(1)}_{\mathbb{M}(r)} \mid X_{\mathbb{O}(r)}^{(1)}, O, R=r)$, it is sufficient to identify the Gibbs factors $p(X^{(1)}_{i} \mid X^{(1)}_{-i}, O, R=r)$ for each index $i \in \mathbb{M}(r)$.
    \item Identification of each Gibbs factor for a pattern $r$ can be established by using one of the independences in \eqref{eq:ri_indep_xi} to show that it is equal to the Gibbs factors of another pattern $\widetilde{r}$ whose joint $p(X^{(1)}, O, R=\widetilde{r})$ is already known to be identified. That is, one way to set up a recursive identification procedure is to process patterns according to a partial order---starting with the complete case pattern---where $\widetilde{r} \prec r$ if identification of one of the Gibbs factors of the extrapolation density of $r$ relies on identifying $p(X^{(1)}, O, R=\widetilde{r})$, and the patterns are incomparable otherwise.
\end{enumerate}

These core ideas of the PM-ID algorithm---in particular the partial order on patterns and borrowing information from earlier patterns to identify individual Gibbs factors of the current pattern---can be encoded succinctly using a \emph{pattern DAG} defined as follows.\footnote{Strictly speaking, per Definition~\ref{def:pattern_dag}, the term pattern DAG here refers to multigraphs and not just simple DAGs. The interpretation of the pattern DAGs we use here is also different than the pattern graphs introduced in \cite{Chen.2022.PatternGraphs}. We elaborate on these differences in Section~\ref{sec:comparison}. }

\begin{definition}[Pattern DAG]
    \label{def:pattern_dag}
    Given an m-graph $\G(V, E)$, a pattern DAG $\mathcal{H}(\widetilde{V}, \widetilde{E})$ is a directed acyclic multigraph (i.e., a directed acyclic graph with the possibility of multiple edges between the same vertices) satisfying the following three properties: (i) The vertex set $\widetilde{V}$ is a subset of all possible missingness patterns, i.e., $\widetilde{V} \subseteq {\cal R}$. (ii) The edge set $\widetilde{E}$ consists of  labeled edges of the form $r \rightarrow^{i} \widetilde{r}$, where $i \in [K]$ is an index and $r, \widetilde{r} \in {\cal R}$ are two distinct missingness patterns. (iii) If there exists an edge $r \rightarrow^{i} \widetilde{r}$ in the pattern DAG ${\cal H}$ then it must be the case that for any distribution $p$ that factorizes according to the m-graph $\G$,
        \begin{equation}
            \label{eq:gibbs_factor_id_pattern_dag}
            p(X_{i}^{(1)} \mid X_{-i}^{(1)}, O, R=r) = p(X_{i}^{(1)} \mid X_{-i}^{(1)}, O, R=\widetilde{r}).
        \end{equation}
\end{definition}
 
Property (iii) tries to link independences in the m-graph $\G$ to identification of Gibbs factors in the extrapolation density of $\widetilde{r}$.  While the equality in \eqref{eq:gibbs_factor_id_pattern_dag} is symmetric, the asymmetry of the edge $r \rightarrow^i \widetilde{r}$ is important: The directed edges in ${\cal H}$ are used to form a partial order of identification ensuring the pattern $r$ is identified before processing $\widetilde{r}$. This ordering then justifies borrowing information from $p(X^{(1)}, O, R=r)$ in order to identify the $i^{th}$ Gibbs factor in the extrapolation density of $p(X^{(1)}, O, R=\widetilde{r})$ based on \eqref{eq:gibbs_factor_id_pattern_dag}.

Given an m-graph $\G$, there are several ways to create a pattern DAG satisfying Definition~\ref{def:pattern_dag}. However, an arbitrary pattern DAG may not necessarily aid in identification using PM factorization. 
We consider a pattern DAG ${\cal H}(\widetilde{V}, \widetilde{E})$ to be \emph{PM-ID} compatible if,
\begin{itemize}
        \item[\textbf{(I1)}] The complete case pattern is present in ${\cal H}$ as a root node (has no incoming edges).
        \item[\textbf{(I2)}] For every $r \in \widetilde{V}$ and every $i \in \mathbb{M}(r)$, there is at least one directed path from the complete case pattern to $r$ such that the last edge on the path is of the form $\circ\rightarrow^i r$.
\end{itemize}
Figures~\ref{fig:id-non-id-pattern-dags}(b) and (c) show  examples of  pattern DAGs that satisfy the properties listed in Definition~\ref{def:pattern_dag} with respect to the m-graph in Figure~\ref{fig:id-non-id-pattern-dags}(a). The pattern DAG in Figure~\ref{fig:id-non-id-pattern-dags}(b) is also PM-ID compatible, as it satisfies (I1, I2). In contrast, the pattern DAG in Figure~\ref{fig:id-non-id-pattern-dags}(c) does not satisfy (I2), as there is no directed path from the complete case pattern to the pattern $00$ that ends with an edge labeled with the missing index $2$.

Intuitively, pattern DAGs that are PM-ID compatible set up a recursive identification argument by linking Gibbs factors of each pattern to ones found in previously identified patterns, with the complete case pattern serving as the base case that is identified simply by consistency. The following theorem formalizes this intuition.\footnote{This result will also prove useful for the extension of PM-ID with positivity violations, where $\widetilde{V} \not= {\cal R}$. We note that we are not the first to provide a constructive argument for full law identification based on PM factorization---\cite{sadinle2016itemwise} provide one for full law identification of the no self-censoring model. However, Theorem~\ref{thm:pattern-dag-id} is more general in the sense that it can be used for constructing the full law using only a subset of  patterns with positive support, as we show in Section~\ref{sec:positivity-violations}.}

\begin{figure}[t]
    \centering
    \subcaptionbox{}{
        \scalebox{0.9}{
        \begin{tikzpicture}
            \def\d{1.5cm}
            \begin{scope}[>=stealth, node distance=2cm]
                \path[->, very thick]
                node[] (x11) {$X^{(1)}_1$}
                node[right of=x11] (x21) {$X^{(1)}_2$}
                
                node[below of=x11] (r1) {$R_1$}
                node[below of=x21] (r2) {$R_2$}
                
                (x21) edge[blue] (r1)
                (x11) edge[blue] (r2)
                
                (x11) edge[red, <->, bend left] (x21)
                (x11) edge[blue, ->] (x21)
                ;
            \end{scope}
        \end{tikzpicture}}
    }
    \hspace{0.1\linewidth}
    \subcaptionbox{}{
        \scalebox{0.8}{
        \begin{tikzpicture}
            \def\d{2.0cm}
            \begin{scope}[>=stealth, node distance=2.0cm]
                \path[->, very thick]
                node[] (11) {$(11)$}
                node[below left of=11] (01) {$(01)$}
                node[below right of=11] (10) {$(10)$}
                node[below right of=01] (00) {$(00)$}

                (11) edge[blue] node [midway, right]{2} (10)
                (11) edge[blue] node [midway, left]{1} (01)
                (01) edge[blue] node [midway, left]{2} (00)
                (10) edge[blue] node [midway, right]{1} (00)
                
                ;
            \end{scope}
        \end{tikzpicture}
        }
    }
    \hspace{0.1\linewidth}
        \subcaptionbox{}{
        \scalebox{0.8}{
        \begin{tikzpicture}
            \def\d{2.0cm}
            \begin{scope}[>=stealth, node distance=2.0cm]
                \path[->, very thick]
                node[] (11) {$(11)$}
                node[below left of=11] (01) {$(01)$}
                node[below right of=11] (10) {$(10)$}
                node[below right of=01] (00) {$(00)$}

                (11) edge[blue] node [midway, right]{2} (10)
                (11) edge[blue] node [midway, left]{1} (01)
                (10) edge[blue] node [midway, right]{1} (00)
                
                ;
            \end{scope}
        \end{tikzpicture}
        }
    }
    \caption{
        Example demonstrating PM-ID compatibility of pattern DAGs: (a) An m-graph $\G$; (b) A pattern DAG defined wrt $\G$ that is PM-ID compatible; (c) A pattern DAG defined wrt $\G$ this is not PM-ID compatible due to the absence of a directed path from the complete case pattern $11$ to the pattern $00$ ending with an edge labeled with the index $2$.
    }
    \label{fig:id-non-id-pattern-dags}
\end{figure}
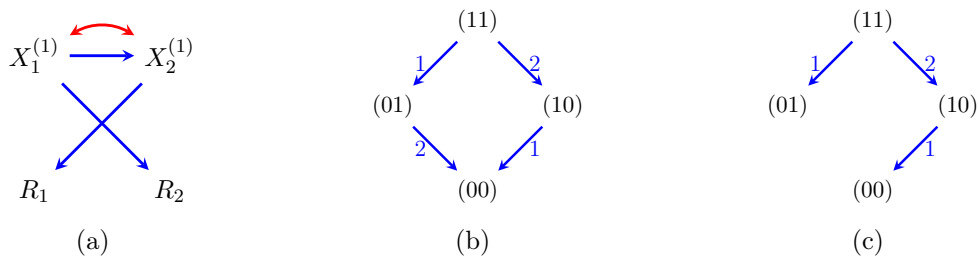

\begin{theorem}[PM-ID Algorithm Using A Pattern DAG] 
    \label{thm:pattern-dag-id}
    Given a missing data model $p(X^{(1)}, O, R, X)$ that satisfies (P1, P2) and factorizes according to an m-graph $\G(V, E)$, let ${\cal H}(\widetilde{V}, \widetilde{E})$ be any PM-ID compatible pattern DAG with respect to $\G$.
    Then the joint distribution $p(X^{(1)}, O, R=r)$ for all patterns $r \in \widetilde{V}$ are identified by processing them according to any valid topological ordering $\prec_{\cal H}$ of ${\cal H}$ as follows.
    \begin{itemize}
        \setlength\itemsep{0em}
        \item The first pattern in $\prec_{\cal H}$ is the complete case pattern, which is identified by consistency as $p(X^{(1)}, O, R=1) = p(X, O, R=1)$.
        \item For subsequent patterns $r\in \prec_{{\cal H}}$, the extrapolation density $p(X_{\mathbb{M}(r)}^{(1)} \mid X_{\mathbb{O}(r)}^{(1)}, O, R=r)$, and hence the joint $p(X^{(1)}, O, R=r)$, is identified as for each $i \in \mathbb{M}(r)$ we have, 
        \begin{equation}
        \label{eq:pm-id-equality}
        p(X^{(1)}_{i} \mid X^{(1)}_{\mathbb{M}(r) \setminus \{i\}}, X^{(1)}_{\mathbb{O}(r)}, O, R=r) 
        =
        p(X^{(1)}_{i} \mid X^{(1)}_{\mathbb{M}(r) \setminus \{i\}}, X^{(1)}_{\mathbb{O}(r)}, O, R=\widetilde{r}),
        \end{equation}
        where $\widetilde{r}$ is any parent of $r$ such that $\widetilde{r} \rightarrow^i r$ exists in ${\cal H}$.
    \end{itemize}
\end{theorem}

\begin{proof}
We prove this result using induction.

\paragraph{Base case} The complete case pattern is indeed the first pattern processed in any valid topological ordering $\prec_{\cal H}$ due to the requirement (I1) and (I2) for ${\cal H}$ to be PM-ID compatible. Further, $p(X^{(1)}, O, R=1)$ is indeed identified by consistency.

\paragraph{Induction hypothesis and induction step} For any subsequent pattern in $r \in \prec_{\cal H}$,  suppose all previous patterns in the ordering have been identified. By requirement (I2) of PM-ID compatibility we know for every Gibbs factor in the extrapolation density $p(X^{(1)}_{\mathbb{M}(r)} \mid X^{(1)}_{\mathbb{O}(r)}, O, R=r)$, there exists at least one parent pattern $\widetilde{r}$ of $r$ such that~\eqref{eq:pm-id-equality} holds. The positivity assumptions (P1, P2) guarantee that these Gibbs factors are well defined. Further, because every $\widetilde{r}$ is a parent of $r$, we also know $\widetilde{r} \prec_{\cal H} r$. That is, by the induction hypothesis and (I2), all Gibbs factors of the extrapolation density of pattern $r$ can be expressed as functions of previously identified patterns. Since all Gibbs factors of the extrapolation density of pattern $r$ are identified, we can conclude that $p(X^{(1)}, O, R=r)$ is identified. The correctness of the base case and this induction step concludes identifiability of all patterns $r \in \widetilde{V}$.
\end{proof}

An immediate corollary of Theorem~\ref{thm:pattern-dag-id} is that given an m-graph $\G(V, E)$, if one can construct a PM-ID compatible pattern DAG ${\cal H}(\widetilde{V}, \widetilde{E})$ that contains all missingness patterns, i.e., $\widetilde{V} = {\cal R}$, then the associated full law $p(X^{(1)}, O, R)$ is identified. There are many ways to construct such a pattern DAG for a given m-graph $\G$. From Section~\ref{sec:constructive_proof:or-id} we have seen that independences of the form \eqref{eq:ri_indep_xi} are sufficient for identifying strictly positive full laws. Thus, a simple approach to pattern DAG construction would focus on using just these constraints.

\begin{algorithm}[t]
    \caption{\textsc{PM-ID Construction}}
    \label{alg:pmid}
    
    \begin{algorithmic}[1]
        \Require m-graph $\G$, set of all possible missingness patterns ${\cal R}$
        \State Initialize a pattern DAG ${\cal H}(\widetilde{V}, \widetilde{E})$ with vertex set $\widetilde{V} = {\cal R}$ and edge set $\widetilde{E} = \emptyset$
        \For {$r \in \mathcal{R}$}
            \State \Comment{For every unobserved variable, check if the Gibbs factors are functions of patterns with one less zero: $p(X_i ^{(1)} \mid X_{-i}^{(1)}, O, R=r)=p(X_i ^{(1)} \mid X_{-i}^{(1)}, O, R_i=1, R_{-i}=r_{-i})$ }
            \For {$i \in \mathbb{M}(r)$}
                \If {$R_i \not\in \mb_\G(X_i^{(1)})$}
                \State Add an edge $\widetilde{r} \rightarrow^{i} r$ to ${\cal H}$, where  $\mathbb{M}(r) \setminus \mathbb{M}(\widetilde{r}) = \{i\}$
                \Else
                \State \Return ``fail''
                \EndIf
            \EndFor
        \EndFor
        \State \Return $\mathcal{H}(\widetilde{V}, \widetilde{E})$
    \end{algorithmic}
\end{algorithm}

Algorithm~\ref{alg:pmid} presents one such method. The PM-ID construction algorithm initializes an empty pattern DAG with vertices $\widetilde{V} = {\cal R}$ and then for every pattern $r \in {\cal R}$ and every index $i\in \mathbb{M}(r)$ attempts to add an edge $\widetilde{r} \rightarrow^i r$ by checking if $\G$ implies $X_i^{(1)} \Perp R_i \mid X_{-i}^{(1)}, R_{-i}, O$, where $\widetilde{r}$ is the pattern in ${\cal R} \setminus \{r\}$ where $i$ is $1$ instead of $0$, i.e., a pattern with one less zero such that  $\mathbb{M}(r) \setminus \mathbb{M}(\widetilde{r}) = \{i\}$. When all these edge additions are successful, the pattern DAG ${\cal H}$ returned by the PM-ID construction algorithm encodes a recursive identification argument via Theorem~\ref{thm:pattern-dag-id} for each pattern $r$ of the full law in terms of the set of all patterns $\widetilde{r}$ that contain precisely one less zero than $r$.

It is easy to check that the PM-ID construction algorithm is successful if and only if the m-graph $\G$ satisfies the structural constraint (S1) due to the check in line~5. Further, when the algorithm returns a pattern DAG ${\cal H}$, it is indeed PM-ID compatible: For (I1), the complete case pattern is guaranteed to be a root node of ${\cal H}$ as a requirement for adding an edge $\widetilde{r} \rightarrow^i r$ in Algorithm~\ref{alg:pmid} is that $\mathbb{M}(r)\setminus \mathbb{M}(\widetilde{r}) = \{i\}$. This is never true for the complete case pattern since $\mathbb{M}(r) = \emptyset$ when $r$ is the complete case pattern. For (I2), let $\prec_{\cal H}$ be any valid topological ordering of the patterns in ${\cal H}$. By (I1), we know the complete case pattern is the first element of $\prec_{\cal H}$ and is the base case pattern that trivially satisfies the criterion in (I2). For any subsequent pattern $r\in \prec_{\cal H}$, suppose all previous patterns satisfy the criterion in (I2). We know from line~6 of Algorithm~\ref{alg:pmid} that an edge  $\widetilde{r} \rightarrow^i r$ is added for every $i \in \mathbb{M}(r)$, where each $\widetilde{r}$ is a pattern with one less zero. Further, it must be the case that $\widetilde{r} \prec_{\cal H} r$ due to these edges. By the induction hypothesis, we conclude that (I2) also holds for $r$, as each edge $\widetilde{r} \rightarrow^i r$ concludes a directed path from the complete case pattern to $r$ for each index $i \in \mathbb{M}(r)$ through a pattern $\widetilde{r}$ that already satisfies (I2). These observations lead to the following results about the PM-ID algorithm.

\begin{corollary}[Soundness of PM-ID]
    \label{cor:soundness-pm-id}
    Suppose a missing data model $p(X^{(1)}, O, R, X)$ satisfies the positivity conditions (P1, P2) and factorizes according to an m-graph $\G(V, E)$ that satisfies (S1). Then the full law $p(X^{(1)}, O, R)$ is identified by applying the PM-ID algorithm described in Theorem~\ref{thm:pattern-dag-id} to the pattern DAG ${\cal H}(\widetilde{V}, \widetilde{E})$ returned by Algorithm~\ref{alg:pmid}.
\end{corollary}
\begin{proof}
    Follows from the fact that ${\cal H}$ is PM-ID compatible and $\widetilde{V} = {\cal R}$.
\end{proof}

\begin{corollary}[Completeness of PM-ID for m-DAGs and m-ADMGs]
    \label{cor:completeness-pm-id}
    Given an m-DAG or m-ADMG $\G(V, E)$, if the pattern DAG construction procedure in Algorithm~\ref{alg:pmid} returns ``fail'', then there exist missing data models $p(X^{(1)}, O, R, X)$ satisfying $(P1, P2)$ and factorizing according to $\G$ whose full laws are provably not identified.
\end{corollary}
\begin{proof}
    Follows from the fact that Algorithm~\ref{alg:pmid} returns ``fail'' if and only if $\G$ does not satisfy (S1), and (S1) is known to be a necessary and sufficient condition for full law identification in m-DAGs and m-ADMGs satisfying (P1, P2) \citep{bhattacharya2019mid, nabi2020mid}.
\end{proof}

Corollaries~\ref{cor:soundness-pm-id} and~\ref{cor:completeness-pm-id} together imply that PM-ID is sound and complete for full law identification in missing data models of m-DAGs and m-ADMGs satisfying positivity, similar to OR-ID. However, in contrast to OR-ID, the PM-ID algorithm also provides explicit identifying functionals. That is, PM-ID has the same identification capabilites as OR-ID, but is also constructive. To demonstrate this, we revisit the previous example, but this time we use the pattern DAG construction in Algorithm~\ref{alg:pmid} and apply PM-ID to it.

\subsubsection{An example of identification using PM-ID With Pattern DAGs}
\label{sec:constructive_proof:pm-id:pattern_dag_example}

\begin{example}

\begin{figure}[t]
    \centering
    \scalebox{0.75}{
        \begin{tikzpicture}
            \def\d{2.0cm}
            \begin{scope}[>=stealth, node distance=2.0cm]
                \path[->, very thick]
                node[] (111) {$(111)$}
                node[below of=111] (101) {$(101)$}
                node[right of=101] (011) {$(011)$}
                node[left of=101] (110) {$(110)$}
                
                node[below of=110] (100) {$(100)$}
                node[below of=101] (010) {$(010)$}
                node[below of=011] (001) {$(001)$}
                node[below of=010] (000) {$(000)$}
                
                (111) edge[blue] node [midway, right]{1} (011)
                (111) edge[blue] node [midway, left]{2} (101)
                (111) edge[blue] node [midway, left]{3} (110)
                (110) edge[blue] node [midway, left]{2} (100)
                (101) edge[blue] node [midway, left]{3} (100)
                (110) edge[blue] node [midway, right]{1} (010)
                (011) edge[blue] node [midway, left]{3} (010)
                (011) edge[blue] node [midway, right]{2} (001)
                (101) edge[blue] node [midway, right]{1} (001)
                (100) edge[blue] node [midway, left]{1} (000)
                (010) edge[blue] node [midway, left]{2} (000)
                (001) edge[blue] node [midway, right]{3} (000)
                ;
            \end{scope}
        \end{tikzpicture}
    }
    \caption{
        The pattern DAG $\mathcal{H}$ obtained by applying the PM-ID pattern DAG construction method in Algorithm~\ref{alg:pmid} to any of the m-graphs shown in Figure~\ref{fig:m-graphs}.
    }
    \label{fig:pattern_dag_strict_positive_example}
\end{figure}
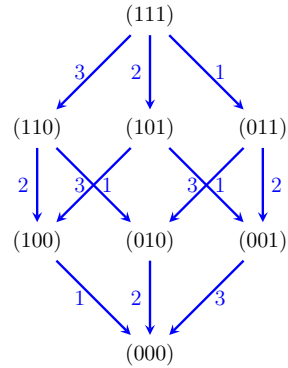

We revisit the m-graphs in Figures~\ref{fig:m-graphs}(a, b, c). Applying Algorithm~\ref{alg:pmid} to any of these m-graphs gives us the pattern DAG in  Figure~\ref{fig:pattern_dag_strict_positive_example}. Explicit identifying functionals encoded by this pattern DAG are obtained according to a topological order $\prec_{\cal H}$ per the PM-ID algorithm in Theorem~\ref{thm:pattern-dag-id} and are shown in display \eqref{eq:pm_id_strict_positive_example}.

\begin{equation}
    \label{eq:pm_id_strict_positive_example}
    \footnotesize
    \begin{aligned}
    p(X^{(1)}, r=111) &= p(X, r=111)
    \\
    p(X^{(1)}, r=110) &= p(X_3^{(1)} \mid X_{12}^{(1)}, r=110) p(X_{12}, r=110)
    \\
    \text{where: } & (111) \rightarrow^3 (110) \; \text{ means } \;
        p(X_3^{(1)} \mid X_{12}^{(1)}, r=110) = p(X_3 \mid X_{12}, r=111)
    \\
    p(X^{(1)}, r=101) &= p(X_2^{(1)} \mid X_{13}^{(1)}, r=101) p(X_{13}, r=101)
    \\
    \text{where: } & (111) \rightarrow^2 (101) \; \text{ means } \;
        p(X_2^{(1)} \mid X_{13}^{(1)}, r=101) = p(X_2 \mid X_{13}, r=111)
    \\
    p(X^{(1)}, r=011) &= p(X_1^{(1)} \mid X_{23}^{(1)}, r=011) p(X_{23}, r=011)
    \\
    \text{where: } & (111) \rightarrow^1 (011) \; \text{ means } \;
        p(X_1^{(1)} \mid X_{23}^{(1)}, r=011) = p(X_1 \mid X_{23}, r=111)
    \\
    p(X^{(1)}, r=100) &= p(X_{23}^{(1)} \mid X_{1}^{(1)}, r=100) p(X_{1}, r=100)
    \\
    \text{where: } &
    \begin{cases}
        (101) \rightarrow^2 (100) \; \text{ means } \;
            p(X_2^{(1)} \mid X_{13}^{(1)}, r=100) = p(X_2^{(1)} \mid X_{13}^{(1)}, r=101)
        \\
        (110) \rightarrow^3 (100) \; \text{ means } \;
            p(X_3^{(1)} \mid X_{12}^{(1)}, r=100) = p(X_3^{(1)} \mid X_{12}^{(1)}, r=110)
        \\
    \end{cases}
    \\
    p(X^{(1)}, r=010) &= p(X_{13}^{(1)} \mid X_{2}^{(1)}, r=010) p(X_{2}, r=010)
    \\
    \text{where: } &
    \begin{cases}
        (110) \rightarrow^1 (010) \; \text{ means } \;
            p(X_1^{(1)} \mid X_{23}^{(1)}, r=010) = p(X_1^{(1)} \mid X_{23}^{(1)}, r=110)
        \\
        (011) \rightarrow^3 (010) \; \text{ means } \;
            p(X_3^{(1)} \mid X_{12}^{(1)}, r=010) = p(X_3^{(1)} \mid X_{12}^{(1)}, r=011)
        \\
    \end{cases}
    \\
    p(X^{(1)}, r=001) &= p(X_{12}^{(1)} \mid X_{3}^{(1)}, r=001) p(X_{3}, r=001)
    \\
    \text{where: } &
    \begin{cases}
        (110) \rightarrow^1 (001) \; \text{ means } \;
            p(X_1^{(1)} \mid X_{23}^{(1)}, r=001) = p(X_1^{(1)} \mid X_{23}^{(1)}, r=110)
        \\
        (011) \rightarrow^2 (001) \; \text{ means } \;
            p(X_2^{(1)} \mid X_{13}^{(1)}, r=001) = p(X_2^{(1)} \mid X_{13}^{(1)}, r=011)
        \\
    \end{cases}
    \\
    p(X^{(1)}, r=000) &:
    \\
    \text{where: } &
    \begin{cases}
        (100) \rightarrow^1 (000) \; \text{ means } \;
            p(X_1^{(1)} \mid X_{23}^{(1)}, r=000) = p(X_1^{(1)} \mid X_{23}^{(1)}, r=100)
        \\
        (010) \rightarrow^2 (000) \; \text{ means } \;
            p(X_2^{(1)} \mid X_{13}^{(1)}, r=000) = p(X_2^{(1)} \mid X_{13}^{(1)}, r=010)
        \\
        (001) \rightarrow^2 (000) \; \text{ means } \;
            p(X_3^{(1)} \mid X_{12}^{(1)}, r=000) = p(X_3^{(1)} \mid X_{12}^{(1)}, r=001)
        \\
    \end{cases}
    \end{aligned}
\end{equation}

\end{example}

\section{Full Law Identification With Positivity Violations}
\label{sec:positivity-violations}

In the previous section we presented a new constructive identification algorithm that works under strict positivity of the missingness mechanism and target law---assumptions (P1) and (P2) respectively. However, in real-world missing data problems, positivity violations are a common practical concern, especially when there are several variables that exhibit missingness. This is due to the exponential increase in the number of potential missingness patterns: with $K$ missing variables, there are $2^K$ possible configurations of missingness indicators. As a result, it is common for many of these patterns to be entirely unobserved in the dataset. In such situations, it becomes unclear whether the full law is still identified and previous identification results for graphical models of missing data do not apply. In this section we extend PM-ID to handle certain kinds of positivity violations. 

\subsection{Formalizing The Relaxed Positivity Assumptions}
\label{sec:positivity-violations:model}

We first formalize the relaxed positivity assumptions.
Most often, data analysts are concerned about certain patterns of missingness not being available in their data. For example, in a three variable problem, an analyst might be concerned that the pattern $101$ never appears in their data, which may correspond to the following positivity violation: $p(R_1=1, R_2=0, R_3=1)=0$.

Such scenarios were not allowed in models satisfying (P1, P2) in the previous section. We relax this as follows. We still require positivity of the target law, but allow certain patterns to be unobserved. That is, we leave (P2) unchanged, and replace (P1) with a new positivity condition that requires positivity on the complete case pattern, but allows other patterns to be have zero support. That is, we introduce a new positivity condition,
\begin{itemize}
    \item[\textbf{(P3)}] The subset of patterns ${\cal R}^+ \subseteq {\cal R}$ that have positive support includes the complete case pattern, and $p(R=r) \geq \epsilon$ if $r \in {\cal R}^+$ and $p(R=r) = 0$ otherwise. Further, ${\cal R}^+$ is compatible with the factorization of $p(X^{(1)}, O, R, X)$ according to $\G$.
\end{itemize}

Before explaining the compatibility clause in (P3), we first confirm that (P1, P3) is a more relaxed set of positivity assumptions than (P1, P2) when ${\cal R}^+$ is a strict subset of ${\cal R}$. To see this, note that when $p(R=r) = 0$ then the joint distribution $p(X^{(1)}, O, R=r) = 0$ for all possible values of $X^{(1)}$ and $O$. This is because $p(R=r) = \sum_{X^{(1)}, O} p(X^{(1)}, O, R=r)$ and since none of the terms in the summation can be negative, they must all be zero when $p(R=r) = 0$. Thus, when some patterns $r$ do not have positive support, i.e., ${\cal R}^+ \subset {\cal R}$, the full law is clearly no longer strictly positive as was the case with (P1, P2). On the other hand, when ${\cal R}^+ = {\cal R}$, the two sets of assumptions are equivalent as expected.

We now explain the compatibility of patterns with the factorization of $\G$ through an example. Consider the m-graph $\G$ in Figure~\ref{fig:id-non-id-pattern-dags}(a). By the factorization of the graph,
\begin{align*}
p(X^{(1)}, R=r) = p(X^{(1)}) \cdot p(R_1 = r_1 \mid X_2^{(1)}) \cdot p(R_2 = r_2 \mid X_1^{(1)}).
\end{align*}
Suppose the pattern $R=10$ does not have positive support, so  $p(X^{(1)}, R=10) = 0$.
Since the target law $p(X^{(1)})$ is required to be positive by (P1), the above factorization then implies that at least one of $p(R_1 = 1 \mid X_2^{(1)})$ or $p(R_2 = 0 \mid X_1^{(1)})$ is equal to zero. If the former is true, then the pattern $R=11$  should also have zero support since $p(R=1 \mid X_2^{(1)})$ also appears in the factorization of $p(X^{(1)}, R=11)$. If the latter is true, then the pattern $R=00$ should also have zero support. Thus, the pattern $R=10$ cannot have zero support in isolation if $p(X^{(1)}, O, R, X)$ factorizes according to the m-graph $\G$.

In the above example there are only two subsets of ${\cal R}$ that satisfy the compatibility requirement of (P3) while also including the complete case pattern: ${\cal R}^+ = \{11, 10\}$ and ${\cal R}^+ = \{11, 01\}$. As negative examples, ${\cal R}^+ = \{10, 00\}$ is compatible but does not contain the complete case, while ${\cal R}^+ = \{11, 00\}$ is not compatible. Moving forward we assume that the observed set of patterns ${\cal R}^+$ is indeed compatible with the given m-graph $\G$, as in our work we assume the data are drawn from a missing data model that factorizes according to $\G$, and $\G$ is correctly specified. Designing tests of model correctness based on compatibility of the observed patterns is an interesting direction of future work; this would extend tests in  \cite{mohan2014testability}, \cite{nabi2022testability}, and \cite{chen2023causal} that are based solely on conditional independence and generalized equality constraints.

\subsection{PM-ID+ For Models With Positivity Violations}
\label{sec:positivity-violations:pm-id-v2}

Under the new positivity assumptions (P1, P3), when ${\cal R}^+ = {\cal R}$ it suffices to use the PM-ID pattern DAG construction in Algorithm~\ref{alg:pmid} along with Theorem~\ref{thm:pattern-dag-id} to identify the full law. However, in general, when some patterns do not have positive support, the identification methods presented in the previous section do not work. This is due to the fact that when $p(R=\widetilde{r}) = 0$ for a certain pattern $\widetilde{r} \in {\cal R}$, then the PM factorization and the Gibbs factors of $p(X^{(1)}, O, R=\widetilde{r})$ are not defined. Thus, any edge $\widetilde{r} \rightarrow^i r$ in the pattern DAG that suggests borrowing information from a pattern $\widetilde{r}$ with no support is also ill-defined.

Similar to our explanation of the PM-ID algorithm, we  present a motivating example that provides concrete intuition for how we can modify the pattern DAG construction of PM-ID to handle positivity violations of the kind allowed by assumptions (P1, P3).

\subsubsection{Motivating example for PM-ID+}
\label{sec:positivity-violations:pm-id-v2:example}

\begin{figure}[t]
    \centering
    \subcaptionbox{}{
        \begin{tikzpicture}
            \def\d{1.5cm}
            \begin{scope}[>=stealth, node distance=1.5cm]
                \path[->, very thick]
                node[] (x11) {$X^{(1)}_1$}
                node[right of=x11] (x21) {$X^{(1)}_2$}
                node[right of=x21] (x31) {$X^{(1)}_3$}
                
                node[below of=x11] (r1) {$R_1$}
                node[below of=x21] (r2) {$R_2$}
                node[below of=x31] (r3) {$R_3$}
                
                (r1) edge[blue] (r2)
                (x21) edge[blue] (r1)
                (x31) edge[blue] (r1)
                (x31) edge[blue] (r2)
                (x11) edge[blue] (r3)
                (x21) edge[blue] (r3)
                
                (x11) edge[blue] (x21)
                (x21) edge[blue] (x31)
                (x11) edge[blue, bend left] (x31)
                
                ;
            \end{scope}
        \end{tikzpicture}
    }
    \hspace{0.05\linewidth}
    \subcaptionbox{}{
        \scalebox{0.75}{
        \begin{tikzpicture}
            \def\d{2.0cm}
            \begin{scope}[>=stealth, node distance=2.0cm]
                \path[->, very thick]
                node[] (111) {$(111)$}
                node[below of=111] (101) {${(101)}$}
                node[right of=101] (011) {$(011)$}
                node[left of=101] (110) {$(110)$}
                
                node[below of=110] (100) {${(100)}$}
                node[below of=101] (010) {$(010)$}
                node[below of=011] (001) {$(001)$}
                node[below of=010] (000) {$(000)$}
                
                (111) edge[blue] node [midway, right]{1} (011)
                (111) edge[blue] node [midway, left]{2} (101)
                (111) edge[blue] node [midway, left]{3} (110)
                (110) edge[blue] node [midway, left]{2} (100)
                (101) edge[blue] node [midway, left]{3} (100)
                (110) edge[blue] node [midway, right]{1} (010)
                (011) edge[blue] node [midway, left]{3} (010)
                (011) edge[blue] node [midway, right]{2} (001)
                (101) edge[blue] node [midway, right]{1} (001)
                (100) edge[blue] node [midway, left]{1} (000)
                (010) edge[blue] node [midway, left]{2} (000)
                (001) edge[blue] node [midway, right]{3} (000)
                
                ;
                \draw[red, very thick] ([xshift=-2pt,yshift=0pt]101.west) -- ([xshift=2pt,yshift=0pt]101.east);
                \draw[red, very thick] ([xshift=-2pt,yshift=0pt]100.west) -- ([xshift=2pt,yshift=0pt]100.east);
            \end{scope}
        \end{tikzpicture}
    }}
    \hspace{0.05\linewidth}
    \subcaptionbox{}{
        \scalebox{0.75}{
        \begin{tikzpicture}
            \def\d{2.0cm}
            \begin{scope}[>=stealth, node distance=2.0cm]
                \path[->, very thick]
                node[] (111) {$(111)$}
                node[below of=111] (101) {}
                node[right of=101] (011) {$(011)$}
                node[left of=101] (110) {$(110)$}

                node[below of=101] (010) {$(010)$}
                node[below of=011] (001) {$(001)$}
                node[below of=010] (000) {$(000)$}
                
                (111) edge[blue] node [midway, right]{1} (011)
                (111) edge[blue] node [midway, left]{3} (110)
                (110) edge[blue] node [midway, right]{1} (010)
                (011) edge[blue] node [midway, left]{3} (010)
                (011) edge[blue] node [midway, left]{1} (001)
                (011) edge[blue, bend left] node [midway, right]{2} (001)
                (111) edge[blue] node [midway, left]{1} (001)
                (010) edge[blue] node [midway, left]{1} (000)
                (010) edge[blue, bend left] node [midway, right]{2} (000)
                (110) edge[blue] node [midway, left]{1} (000)
                (001) edge[blue] node [midway, right]{3} (000)
                
                ;
            \end{scope}
        \end{tikzpicture}
        }
    }
    \caption{
        (a) An m-graph $\G$ used to illustrate pattern mixture identification with positivity violations on the patterns $R=101$ and $R=100$. (b) Pattern DAG obtained by applying the PM-ID construction algorithm. The red strikethrough indicates patterns with no support. (c) Pattern DAG obtained by applying the PM-ID+ construction algorithm.
    }
    \label{fig:pattern_dag_example_non_positive_1}
\end{figure}
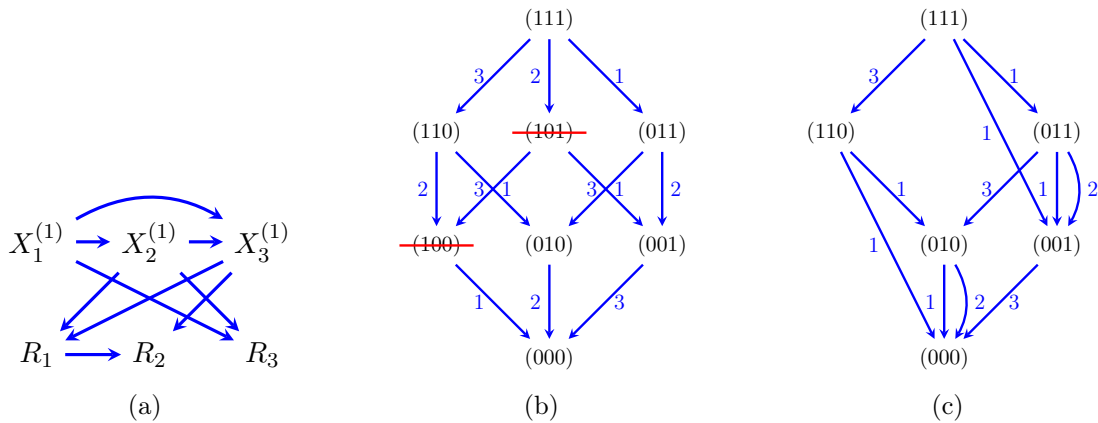

\begin{example}

Let $p(X^{(1)}, O, R, X)$ be a missing data model that factorizes according to the m-graph $\G$  in Figure~\ref{fig:pattern_dag_example_non_positive_1}(a) and has support for patterns ${\cal R}^+ = \{111, 011, 110, 010, 001, 000\}$. That is, the patterns $\{101, 100\}$ have no support, so $p(R=101) = 0$ and $p(R=100) = 0$. Clearly this model satisfies assumptions (P1, P3)\footnote{We manually checked that this set of patterns is also compatible with the factorization of $\G$.} but not (P1, P2).

Say we naively apply the PM-ID pattern DAG construction algorithm to $\G$; we would then obtain the pattern DAG shown in Figure~\ref{fig:pattern_dag_example_non_positive_1}(b). The red strikethroughs in the figure are used to denote the patterns that have no support in the model.

\begin{enumerate}
    \item Based on this pattern DAG, we are still able to use the same arguments as before to identify the patterns $111, 110, 011$, and $010$, as the identification of these patterns do not rely on borrowing information from the patterns with no support.
    \item Identification for patterns with no support is also simple: since $p(R=101)=0$ and $p(R=100)=0$, we have that $p(X^{(1)}, O, R=101) = 0$ and $p(X^{(1)}, O, R=100) = 0$.
    \item However, the patterns $001$ and $000$, whose identification arguments use information from patterns with no support, present a problem. Consider $R=001$ for example: The edges $(101)\rightarrow^1 (001)$ and $(011) \rightarrow^2 (001)$ in Figure~\ref{fig:pattern_dag_example_non_positive_1}(b) suggest identification of the Gibbs factors of the extrapolation density of $p(X^{(1)}, O, R=001)$ via the two  equalities,
    \begin{align*}
         p(X^{(1)}_{1} \mid X^{(1)}_{23}, R=001) &= p(X^{(1)}_{1} \mid X^{(1)}_{23}, R=101), \text{ and } \\
         p(X^{(1)}_{2} \mid X^{(1)}_{13}, R=001) &= p(X^{(1)}_{2} \mid X^{(1)}_{13}, R=011).
    \end{align*}
    The problem lies in the fact that $p(X^{(1)}_{1} \mid X^{(1)}_{23}, R=101)$ is undefined, since the full law at $101$ is zero, so the first equality cannot be used for identification. A similar problem occurs for $R=000$ when attempting to use the edge $(100) \rightarrow^1 (000)$ for identification.
    \item Recall that while the PM-ID pattern DAG construction only uses constraints of the form $X_i^{(1)} \Perp R_i \mid \mb_\G(X_i^{(1)})$, the m-graph $\G$ may imply additional constraints based on the local Markov property~\eqref{eq:local-property}.  Thus, in order to patch the issue of the edge $(101)\rightarrow^1 (001)$ being ill-defined due to a positivity violation in the pattern $101$, we may employ additional constraints to tether the first Gibbs factor to a different well-defined conditional. By the local Markov property, $\G$ also implies $X_1^{(1)} \Perp R_{12} \mid X^{(1)}_{23}, R_{3}$, which leads to
        \begin{equation*}
            p(X^{(1)}_{1} \mid X^{(1)}_{23}, R=001) = p(X^{(1)}_{1} \mid X^{(1)}_{23}, R=111).
        \end{equation*}
        The right hand side is a conditional of the complete case joint distribution, and hence is well-defined due to the positive support on $111$. Moreover, this conditional is also identified by consistency. This concludes that the full law at $001$ is identified by replacing the ill-defined edge $(101 )\rightarrow^1 (001)$ with $(111) \rightarrow^1 (001)$ instead. Note the local Markov property also allows us to add the edge $(011) \rightarrow^1 (001)$, offering multiple ways of identifying this Gibbs factor and possibly increasing data efficiency in estimation.

    \item We can perform a similar patch for the pattern $R=000$. By the local Markov property of $\G$ we have,
        \begin{equation*}
            \begin{aligned}
                p(X^{(1)}_{1} \mid X^{(1)}_{23}, R=000) &= p(X^{(1)}_{1} \mid X^{(1)}_{23}, R=110), \: \: \text{and}
                \\
                p(X^{(1)}_{1} \mid X^{(1)}_{23}, R=000) &= p(X^{(1)}_{1} \mid X^{(1)}_{23}, R=010).
            \end{aligned}
        \end{equation*}
        That is, the ill-defined edge $(100)\rightarrow^1 (000)$ can be replaced by the edges $(110)\rightarrow^1 (000)$ and $(010) \rightarrow^1 (000)$, i.e., edges from patterns with positive support and that can be identified prior to $000$. Thus, we conclude that $p(X^{(1)}, O, R=000)$ is also identified.
    \item The full pattern DAG obtained by replacing ill-defined edges due to positivity violations to edges that only use patterns with positive support is shown in Figure~\ref{fig:pattern_dag_example_non_positive_1}(c). The goal of PM-ID+ in the next subsection will be to produce such pattern DAGs.
\end{enumerate}

\end{example}

\subsubsection{The PM-ID+ Algorithm}
\label{sec:positivity-violations:pm-id-v2:theorem}

The example in the previous subsection provides a few key insights that we use to phrase an extension of the PM-ID algorithm for models with positivity violations.
\begin{enumerate}
    \item Identification of patterns with no support is trivial as $p(X^{(1)}, O, R=r) = 0$ for any pattern $r \in {\cal R}\setminus {\cal R}^+$. So we only need to focus on identification of patterns in ${\cal R}^+$.
    \item For patterns $r \in {\cal R}^+$ with positive support, the main recursive identification argument from Theorem~\ref{thm:pattern-dag-id} can be preserved if we are able to produce a PM-ID compatible pattern DAG ${\cal H}(\widetilde{V}, \widetilde{E})$---a pattern DAG that satisfies (I1, I2)---such that $\widetilde{V} = {\cal R}^+$.
    \item The construction of such a pattern DAG should use additional independences  implied by the local Markov property of the m-graph $\G$. In particular, independences of the form $X_i^{(1)} \Perp R_j \mid \mb_\G(X_i^{(1)})$ for $j\not=i$, in addition to the independences $X_i^{(1)} \Perp R_i \mid \mb_\G(X_i^{(1)})$ used in the PM-ID construction algorithm in Algorithm~\ref{alg:pmid}.
   
\end{enumerate}

Based on these insights, we propose the PM-ID+ pattern DAG construction strategy in Algorithm~\ref{alg:pmid-plus}. We use  $\text{ID} \subseteq {\cal R}^+$ to denote the set of patterns without positive support that have been successfully identified at any given step of the algorithm. $\text{ID}$ is initialized to be $\emptyset$, and at each iteration of the identification process, this set is potentially expanded to include newly identified patterns. By (P3), the complete case is in ${\cal R}^+$. Thus on the first iteration of the algorithm we are guaranteed to get identification of the complete case pattern  by consistency, and so the complete case pattern is immediately added to $\text{ID}$. This also guarantees that if a pattern DAG is successfully produced by the PM-ID+ construction, it satisfies the condition (I1) that the complete case pattern is a root node of ${\cal H}$.

At each iteration, the algorithm loops through all patterns with positive support whose identifiability has not yet been determined, i.e., each $r \in {\cal R}^+ \setminus \text{ID}$, and tries to establish their identifiability based on previously identified patterns with positive support.
Concretely, for every index corresponding to actually missing variable $i \in \mathbb{M}(r)$, we list all other patterns $\widetilde{r} \in \text{ID}$ such that the equality in \eqref{eq:gibbs_factor_id_pattern_dag} holds, implying that the $i^{th}$ Gibbs factor of the extrapolation density of pattern $r$ can be re-expressed as the $i^{th}$ Gibbs factor of an already identified pattern $\widetilde{r}$. By the local Markov property, the set of all such patterns is computed in lines 9 and 10 of Algorithm~\ref{alg:pmid-plus} as,
\begin{equation}
    {\cal Z}(r, i) = \{\widetilde{r} \in \text{ID} \mid R_j \not\in \mb_\G(X_i^{(1)}) \text{ for all } j \in \diff(r, \widetilde{r} )\},
\end{equation}
where $\diff(r, \widetilde{r}) \coloneqq \{i \in [K] \mid r_i \not= \widetilde{r}_i\}$. If we are able to link every Gibbs factor of the extrapolation density of pattern $r$ to at least one already identified pattern, then we know the joint $p(X^{(1)}, O, R=r)$ is identifiable and we can append edges of the form $\widetilde{r} \rightarrow^i r$ to the pattern DAG that satisfy the PM-ID compatibility criterion (I2) for pattern $r$ by induction. This addition of edges occurs in line 13 of Algorithm~\ref{alg:pmid-plus}.

Algorithm~\ref{alg:pmid-plus} iteratively updates the set of identified patterns in this manner, repeating the procedure until no further updates are possible. This can occur when $\text{ID} \subset {\cal R}^+$ is not identified, but no new pattern can be identified by our procedure. In this case the algorithm returns ``fail''. The other case is when $\text{ID} = {\cal R}^+$, so all patterns with positive support have been identified and the PM-ID+ construction returns a PM-ID compatible pattern DAG ${\cal H}(\widetilde{V}, \widetilde{E})$ with $\widetilde{V} = {\cal R}^+$. As an example of success, the PM-ID+ construction algorithm yields the pattern DAG in Figure~\ref{fig:pattern_dag_example_non_positive_1}(c) for the m-graph in Figure~\ref{fig:pattern_dag_example_non_positive_1}(a) with positivity violations on patterns $101$ and $100$. The following  result formalizes the soundness of PM-ID+.

\begin{algorithm}[t]
    \caption{\textsc{PM-ID$^+$ Construction}}
    \label{alg:pmid-plus}
    
    \begin{algorithmic}[1]
    
        \Require m-graph $\G(V, E)$, patterns with positive support ${\cal R}^+ = \{r \in {\cal R} \mid p(R=r) > 0\} $
        \State Initialize a pattern DAG ${\cal H}(\widetilde{V}, \widetilde{E})$ with vertex set $\widetilde{V} = {\cal R}^+$ and edge set $\widetilde{E} = \emptyset$
        \State \Comment{Create variables to keep track of all patterns identified and new patterns identified in an iteration (initialize the latter with the complete case pattern)}
        \State $\text{ID} \gets \emptyset$ and $\Delta \gets \{r=1\}$ 
        \While {$\Delta \not= \emptyset$}
        \State $\text{ID} \gets \text{ID} \cup \Delta$ and $\Delta \gets \emptyset$
        \For {$r \in {\cal R}^+ \setminus \text{ID} $} \label{alg:iter-patterns}
            
            \State \Comment{Check if Gibbs factors $p(X_i ^{(1)} \mid X_{-i}^{(1)}, O, R=r)$ of every unobserved variable can be expressed as a function of at least one pattern that has been identified so far}
            \For {$i \in \mathbb{M}(r)$}
            \State Let $\diff(r, \widetilde{r}) \coloneqq \{i \in [K] \mid r_i \not= \widetilde{r}_i\}$
            \State ${\cal Z}(r, i) \gets \{\widetilde{r} \in \text{ID} \mid R_j \not\in \mb_\G(X_i^{(1)}) \text{ for all } j \in \diff(r, \widetilde{r} )\}$
            \If {${\cal Z}(r, i) =\emptyset$} {go to the next pattern in line~\ref{alg:iter-patterns}}
            \EndIf
            \EndFor
            \State \Comment{Here we know that every Gibbs factor can be identified as a function of one or more identified patterns; add edges for this and mark the pattern as identified}
            \State {\textbf{for} all $i \in \mathbb{M}(r)$ and \textbf{for} all $\widetilde{r} \in {\cal Z}(r, i)$} add an edge $\widetilde{r} \to^i r$ to ${\cal H}$
            \State $\Delta \gets \Delta \cup \{r\}$
        \EndFor
        \If {$\Delta = \emptyset$}
        \State \Return ``fail'' \EndIf
        \EndWhile
        \State \Return $\mathcal{H}(\widetilde{V}, \widetilde{E})$

    \end{algorithmic}
\end{algorithm}

\begin{corollary}[Soundness of PM-ID+]
    \label{cor:soundness-pm-id-plus}
    Suppose a missing data model $p(X^{(1)}, O, R, X)$ satisfies the positivity conditions (P1, P3) and factorizes according to an m-graph $\G(V, E)$. Then the full law $p(X^{(1)}, O, R)$ is identified if the PM-ID+ pattern DAG construction in Algorithm~\ref{alg:pmid-plus} succeeds in returning a pattern DAG ${\cal H}(\widetilde{V}, \widetilde{E})$, and is identified as follows.
    \begin{itemize}
        \item For patterns $r \in {\cal R}\setminus {\cal R}^+$, we have that $p(X^{(1)}, O, R=r) = 0$.
        \item For patterns $r\in {\cal R}^+$, we apply the PM-ID algorithm in Theorem~\ref{thm:pattern-dag-id} to ${\cal H}(\widetilde{V}, \widetilde{E})$.
    \end{itemize}
\end{corollary}
\begin{proof}
    Identification of patterns $r\in {\cal R} \setminus {\cal R}^+$ is trivial. Identification of all patterns $r\in {\cal R}^+$ follows from the fact that ${\cal H}$ is PM-ID compatible with $\widetilde{V} = {\cal R}^+$ and the assumptions (P1, P2) can be swapped with (P1, P3) when all patterns in ${\cal H}$ have positive support.
\end{proof}

Note that we make no claim of completeness for PM-ID+. That is, there may exist missing data models that satisfy (P1, P3) and factorize according to an m-graph, such that the PM-ID+ construction returns ``fail'', but the full law is in fact identified. Nonetheless, PM-ID+ covers a large set of MNAR models that may exhibit positivity violations, and that to our knowledge, have not been considered in prior work on m-graphs.

We provide one more example of the PM-ID+ pattern DAG construction algorithm that demonstrates an interesting aspect of the algorithm---while the PM-ID construction makes it so that a pattern $r$ can only borrowing information from another pattern $\widetilde{r}$ with strictly fewer missing variables (i.e., $\mathbb{M}(r) \supset \mathbb{M}(\widetilde{r})$), the PM-ID+ construction only requires that $\widetilde{r}$ be identifiable prior to $r$, regardless of the relation between $\mathbb{M}(r)$ and $\mathbb{M}(\widetilde{r})$.


\begin{example}
Let $p(X^{(1)}, O, R, X)$ be a missing data model that factorizes according to the m-graph in Figure~\ref{fig:pattern_dag_example_non_positive_2}(a) and ${\cal R}^+ = \{111, 110, 011, 100, 001\}$. Thus, ${\cal R} \setminus {\cal R}^+ = \{101, 010, 000\}$ have no support. We now list each iteration of the PM-ID+ construction in Algorithm~\ref{alg:pmid-plus} that leads to the pattern DAG in Figure~\ref{fig:pattern_dag_example_non_positive_2}(b).
\begin{enumerate}
    \item $\text{ID} = \emptyset$ and the newly identified patterns $\Delta = \{111\}$
    \item $\text{ID} = \{111\}$ and the newly identified patterns $\Delta = \{110, 011\}$, as ${\cal Z}(110, 3) = \{111\}$ and ${\cal Z}(011, 1) = \{111\}$. That is, borrowing information from  just the complete pattern is sufficient for identifying these two patterns.
    \item $\text{ID} = \{111, 110, 011\}$ and the newly identified patterns $\Delta = \{001\}$, as ${\cal Z}(001, 1) = \{ 111, 011\}$ and ${\cal Z}(001, 2) = \{ 111, 011\}$. That is, we can identify the Gibbs factors in the extrapolation density of $p(X^{(1)}, O, R=r)$ by borrowing information from both the $111$ and $011$ patterns.
    \item $\text{ID} = \{111, 110, 011, 001\}$ and the newly identified patterns $\Delta = \{100\}$, as ${\cal Z}(100, 2) = \{ 110\}$ and ${\cal Z}(100, 3) = \{ 001\}$. This leads to the addition of an edge $(001) \rightarrow^3 (100)$ in ${\cal H}$ that represents an interesting case where, PM-ID+  borrows information from a pattern where the observed variables are not necessarily a strict superset of those in the current pattern.
    \item Finally $\text{ID} = {\cal R}^+$ so the full law is identified, and the algorithm returns the pattern DAG ${\cal H}$ shown in Figure~\ref{fig:pattern_dag_example_non_positive_2}(b).
\end{enumerate}

\begin{figure}[t]
    \centering
    \subcaptionbox{}{
        \begin{tikzpicture}
            \def\d{1.5cm}
            \begin{scope}[>=stealth, node distance=1.5cm]
                \path[->, very thick]
                node[] (x11) {$X^{(1)}_1$}
                node[right of=x11] (x21) {$X^{(1)}_2$}
                node[right of=x21] (x31) {$X^{(1)}_3$}
                
                node[below of=x11] (r1) {$R_1$}
                node[below of=x21] (r2) {$R_2$}
                node[below of=x31] (r3) {$R_3$}
                
                (r2) edge[blue] (r1)
                (r3) edge[blue, bend left] (r1)
                (x31) edge[blue] (r2)
                (x11) edge[blue] (r3)
                (x21) edge[blue] (r3)
                
                (x11) edge[blue] (x21)
                (x21) edge[blue] (x31)
                (x11) edge[blue, bend left] (x31)
                
                ;
            \end{scope}
        \end{tikzpicture}
    }
    \hspace{0.1\linewidth}
    \subcaptionbox{}{
        \scalebox{0.8}{
        \begin{tikzpicture}
            \def\d{2.0cm}
            \begin{scope}[>=stealth, node distance=2.0cm]
                \path[->, very thick]
                node[] (111) {$(111)$}
                node[below of=111] (101) {}
                node[right of=101] (011) {$(011)$}
                node[left of=101] (110) {$(110)$}
                
                node[below of=110] (100) {$(100)$}
                node[below of=101] (010) {}
                node[below of=011] (001) {$(001)$}
                
                (111) edge[blue] node [midway, right]{1} (011)
                (111) edge[blue] node [midway, left]{3} (110)
                (011) edge[blue] node [midway, right]{1} (001)
                (011) edge[blue, bend left] node [midway, right]{2} (001)
                (111) edge[blue, bend left=90] node [midway, right]{1} (001)
                (111) edge[blue, bend right] node [midway, right]{2} (001)
                (110) edge[blue] node [midway, left]{2} (100)
                (001) edge[blue] node [midway, above]{3} (100)
                
                ;
                
            \end{scope}
        \end{tikzpicture}
        }
    }
    \caption{
        (a) An m-graph used to demonstrate that PM-ID+ can borrow information from patterns where the sets $\mathbb{M}(r)$ and $\mathbb{M}(\widetilde{r})$ are incomparable. (b) The pattern DAG produced by the PM-ID+ construction, where the edge $001 \rightarrow^3 100$ illustates this point.
    }
    \label{fig:pattern_dag_example_non_positive_2}
\end{figure}
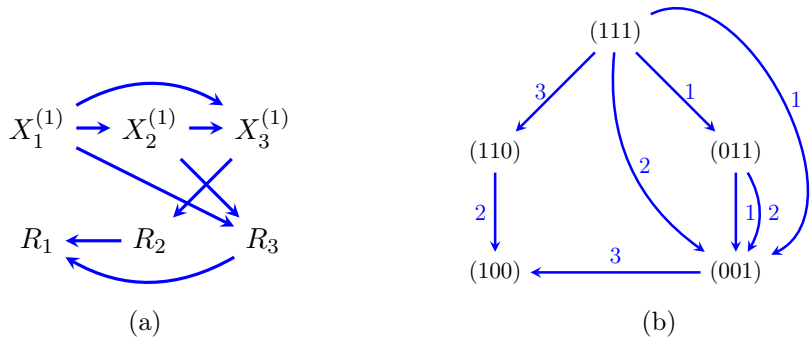

\end{example}

\subsection{Relative merits of PM-ID and PM-ID+}

\begin{figure}[t]
    \centering
    \subcaptionbox{}{
        \begin{tikzpicture}
            \def\d{1.5cm}
            \begin{scope}[>=stealth, node distance=1.5cm]
                \path[->, very thick]
                node[] (o) {$O$}
                node[right of=o] (x11) {$X^{(1)}_1$}
                node[right of=x11] (x21) {$X^{(1)}_2$}
                node[below of=x11] (r1) {$R_1$}
                node[below of=x21] (r2) {$R_2$}
                
                (o) edge[blue] (x11)
                (x11) edge[blue] (x21)
                (o) edge[blue, bend left] (x21)
                (o) edge[blue] (r1)
                (o) edge[blue] (r2)
                (r1) edge[-, brown] (r2)
                ;
            \end{scope}
        \end{tikzpicture}
    }
    \hspace{0.05\linewidth}
    \subcaptionbox{}{
        \scalebox{0.7}{
            \begin{tikzpicture}
                \def\d{2.0cm}
                \begin{scope}[>=stealth, node distance=2.0cm]
                    \path[->, very thick]
                    node[] (11) {$(11)$}
                    node[below of=11, xshift=-\d] (10) {$(10)$}
                    node[right of=10, xshift=\d] (01) {$(01)$}
                    node[below of=01, xshift=-\d] (00) {$(00)$}
                    
                    (11) edge[blue] node [midway, above]{2} (10)
                    (11) edge[blue] node [midway, above]{1} (01)
                    (10) edge[blue] node [midway, below]{1} (00)
                    (01) edge[blue] node [midway, below]{2} (00)
                    ;
                \end{scope}
            \end{tikzpicture}
        }
    }
    \hspace{0.05\linewidth} 
    \subcaptionbox{}{
        \scalebox{0.7}{
            \begin{tikzpicture}
                \def\d{2.0cm}
                \begin{scope}[>=stealth, node distance=2.0cm]
                    \path[->, very thick]
                    node[] (11) {$(11)$}
                    node[below of=11, xshift=-\d] (10) {$(10)$}
                    node[right of=10, xshift=\d] (01) {$(01)$}
                    node[below of=01, xshift=-\d] (00) {$(00)$}
                    
                    (11) edge[blue] node [midway, above]{2} (10)
                    (11) edge[blue] node [midway, above]{1} (01)
                    (10) edge[blue, bend left] node [midway, above]{2} (01)
                    (11) edge[blue, bend right=20] node [midway, right]{1} (00)
                    (11) edge[blue, bend left=20] node [midway, right]{2} (00)
                    (10) edge[blue] node [midway, below]{1} (00)
                    (10) edge[blue, bend right=25] node [midway, below]{2} (00)
                    (01) edge[blue, bend left=25] node [midway, below]{1} (00)
                    (01) edge[blue] node [midway, below]{2} (00)
                    ;
                \end{scope}
            \end{tikzpicture}
        }
    }
    \caption{
        (a) A sparse m-graph used to demonstrate the relative data efficiency of PM-ID+ and the robustness to graph misspecification of PM-ID. (b) The pattern DAG obtained using the PM-ID construction. (c) The pattern DAG obtained using the PM-ID+ construction.
    }
    \label{fig:pmid_pmidplus_compare}
\end{figure}
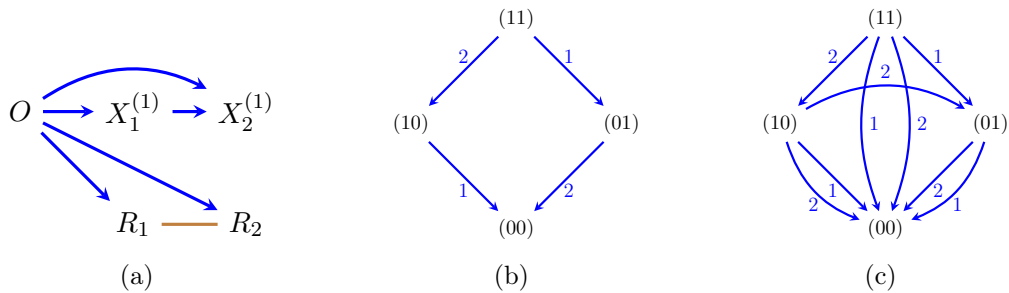

Before we present an imputation algorithm based on these identification results, we briefly comment on the relative tradeoffs between using PM-ID and PM-ID+ for the identification step preceding imputation of missing values. 

The obvious advantage of PM-ID+ is that it can identify missing data models under a weaker set of positivity assumptions. That is, the set of missing data models that can be identified by PM-ID+ is strictly larger than those identifiable by PM-ID. Further, since PM-ID+ uses additional constraints encoded by the m-graph, it can often lead to pattern DAGs that make use of more rows of data than PM-ID, even when the full law is strictly positive. As an example, consider the m-graph $\G$ in Figure~\ref{fig:pmid_pmidplus_compare}(a) and the  pattern DAGs produced in Figures~\ref{fig:pmid_pmidplus_compare}(b) and (c) produced by the PM-ID and PM-ID+ constructions respectively. The PM-ID+ pattern DAG makes use of more constraints encoded by the m-graph, e.g., when identifying the pattern $00$, it uses information from all of the other patterns, while PM-ID does not. This may affect the efficiency of downstream imputation by dictating what rows of data can be used to fit each Gibbs factor in the extrapolation density. For example, for the pattern $00$  with Gibbs factors $p(X_1^{(1)} \mid X_2^{(1)}, R=00)$ and $p(X_2^{(1)} \mid X_1^{(1)}, R=00)$, PM-ID+ proposes each of these Gibbs factors can be fit using data from any other row, while PM-ID proposes the first Gibbs factor can be fit only using rows of data with the pattern $10$ and the second can be fit only using rows with the pattern $01$. An interesting question for future is whether this also leads to greater precision in downstream estimates of target parameters.

Next, we provide a brief analysis on the computational complexity  PM-ID and PM-ID+, as a function of the number of missing variables $K$ and missing patterns with positive support $|{\cal R}|^+$. That is, for this analysis, we consider sets like the Markov blanket of variables to have been pre-computed and stored, as these operations are fast and can be performed in linear time in number of edges and vertices of the graph using depth first search.

For PM-ID, we require all possible patterns to have positive support. In the outer loop we iterate over these patterns and in the inner loop, we iterate through every missing index in these patterns. That is, the outer loop executes $|{\cal R}|^+ = 2^K$ times and the inner loop executes at most $K$ times (for the pattern with all zeros), where $K$ is the number of potentially missing variables and $\cal R^{+}$ is the number of positivity patterns. This gives us a total computational complexity of ${\cal O}(K2^K)$ for PM-ID. The following lemma provides the computational complexity of PM-ID+.
\begin{lemma}
    PM-ID+ has worst case time complexity $\mathcal{O}(K|\mathcal{R}^{+}|^3)$, where $K$ is the number of missing variables and $|{\cal R}|^+$ is the number of missingness patterns with support.
    \label{lem:complexity-pm-id-plus}
\end{lemma}
\begin{proof}
    In the worst case, we identify only one new pattern with positive support on each iteration of PM-ID+. That is, the outermost while loop executes at most $|{\cal R}^+|-1$ times. Similarly, the inner for loop on line 6 of the algorithm executes at most $|{\cal R}^+|-1$ on the first iteration when only the complete case pattern is known to be identified. The block consisting of lines 8-13 and the loops in line 13 both perform at most $K(|{\cal R}^+|-1)$ steps---when the pattern $r$ is the pattern with all zeros and every other pattern with positive support has been identified before it. This gives us the overall complexity of ${\cal O}(K|{\cal R}^+|^3)$.
\end{proof}
Note that when all patterns have positive support, i.e., $|{\cal R}^+| = 2^K$, Lemma~\ref{lem:complexity-pm-id-plus} implies that PM-ID+ has worse computational complexity than PM-ID, as in this case ${\cal O}(K|{\cal R}^+|^3) = {\cal O}(K8^K)$. 
However, in most practical applications it is likely to be the case that $|{\cal R}^+| \ll 2^K$ and the number of positive patterns is, in fact, some function of the number of rows of data $N$. For example, in a problem with twenty missing variables (i.e., $K=20$) there are over 1 million possible missingness patterns. However, most available datasets are likely to have number of rows $N$ much smaller than that.
Thus, the number of patterns with support $|\mathcal{R}^{+}|$ is typically an order of magnitude smaller than $2^K$. For example, $|\mathcal{R}^{+}|$ may only be a polynomial (or perhaps even logarithmic) function of $K$ and $N$. Therefore, in such cases, the time complexity of PM-ID+ is also polynomial, rather than exponential, with respect to the number of missing variables $K$ and number of rows $N$. We also see this empirically in our simulations in Section~\ref{sec:experiments}, where PM-ID+ and the imputation algorithm MISPR can finish running relatively quickly when the number of patterns with positive support grows slowly with respect to the size of the dataset.


While it might seem like PM-ID+ is better than PM-ID in most respects, the simplicity of the PM-ID construction does have one significant benefit---it can protect against model misspecification when there is significant uncertainty about the exact structure of the m-graph $\G$. If, for example, the m-graph in Figure~\ref{fig:pmid_pmidplus_compare}(a) is incorrect and there in fact exists an edge $X_2^{(1)}\rightarrow R_1$. Then the pattern DAG obtained from the PM-ID construction remains valid. However, the pattern DAG from PM-ID+   includes edges that are incorrect due to misspecification of the m-graph, e.g., the edge $(11)\rightarrow^2 (00)$  is incorrect as $p(X_2^{(1)} \mid X_1^{(2)}, R=00)\not= p(X_2^{(1)} \mid X_1^{(2)}, R=11)$ when $X_2^{(1)} \rightarrow R_1$ is present. Thus, in cases of model uncertainty, PM-ID might be preferrable to PM-ID+. In fact, the analyst need not even specify a full m-graph for PM-ID, but rather only justify the graphical condition (S1).

Note that neither PM-ID nor PM-ID+ support positivity violations on the complete case pattern. One way to handle this is to also incorporate independences of the form $X_i^{(1)} \Perp X_j^{(1)} \mid \mb_\G(X_i^{(1)})$ into the pattern DAG construction process, i.e., use constraints on the target law $p(X^{(1)}, O)$ implied by the local Markov property of the m-graph $\G$. We provide an example of how to do this in Appendix~\ref{app:no_complete_cases}. However, we do not expand on such techniques further in this paper, as it is often desirable to assume the least restrictive model possible for the target law to prevent bias in downstream analysis.

\section{MISPR: Multivariate Imputation via Supported Pattern Recursion}
\label{sec:imputation_algorithm}

Having described our identification algorithms, we are finally ready to describe our imputation algorithm that uses either PM-ID or PM-ID+ as a pre-processing step to ensure the full law is identified and to direct the imputation procedure itself based on the pattern DAG.

Data analysts often only have access to a finite data table $\mathbb{D}_{N\times(2K+J)}$, a table with $N$ rows of data where the $n^{th}$ row of data is a vector $[x_1^n, \dots, x_K^n, r_1^n, \dots, r_K^n, o_1^n, \dots, o_J^n]$, where the values $o_j^n$ are always observed, while the values $x_k^n$ could be missing (denoted as ``?'' or ``NA''), as indicated by the value of the corresponding missingness indicator $r_k^n$. The goal of this section is to provide an imputation algorithm which fills in missing values in such a data table. Fortunately, since PM-ID and PM-ID+ are constructive algorithms, it is relatively easy to translate the identification strategy that operates on a distribution level to estimation and imputation algorithms with finite data.

To see this, it is insightful to note that knowing the extrapolation density for pattern $r$ is sufficient to perform imputation for all data rows matching this missing pattern, independent of missing data model considered. Concretely, for each data row $D^{n}=(X_{\mathbb{M}(r^n)} = ?, X_{\mathbb{O}(r^n)} = x_{\mathbb{O}(r^n)}^{n}, O=o^n, R=r^n)$ such that $r^n$ matches a certain pattern $r$, the missing values for $X_{\mathbb{M}(r^n)}$ can be imputed by sampling values from the extrapolation density $x^{(1)}_{\mathbb{M}(r)} \sim p(X^{(1)}_{\mathbb{M}(r)} \mid X_{\mathbb{O}(r)} = x_{\mathbb{O}(r)}^{n}, O = o^{n}, R=r)$. Hence, knowing the extrapolation densities at all patterns present in the data table is sufficient for imputation of that table. Further, we can use PM-ID or PM-ID+, as appropriate, to obtain explicit expressions of the Gibbs factors of the extrapolation densities of all patterns with positive support. Hence, imputation for a data row with pattern $r$ can proceed by Gibbs sampling using the Gibbs factors for the extrapolation density for this pattern.

\begin{algorithm}[t]
    \caption{Mutivariate Imputation via Supported Pattern Recursion (MISPR)}
    \label{alg:imputation_via_pattern_dag}
    
    \begin{algorithmic}[1]
    
        \Require m-graph $\G$, Data $\mathbb{D}_{N \times (2K+J)}$ where $D^n = [x_1^n, \dots,x_K^n, r_1^n, \dots, r_K^n, o_1^n, \dots, o_J^n]$
        
        \State ${\cal R}^+ \gets \text{all unique missingness patterns } r \text{ in } \mathbb{D}$
        \If{ ${\cal R} = {\cal R}^+$}
        \State ${\cal H}(\widetilde{V}, \widetilde{E}) \gets \textsc{PM-ID Construction}(\G, {\cal R}^+)$
        \Else \State ${\cal H}(\widetilde{V}, \widetilde{E}) \gets \textsc{PM-ID+ Construction}(\G, {\cal R}^+)$
        \EndIf
        \If{${\cal H} = \text{``fail''}$ } {\textbf{return} ``fail''}
        \EndIf
        \vspace{0.5em}
        \For {$r \in \prec_{{\cal H}}$}
            \State \Comment{Fit Gibbs factors using data rows we can borrow information from according to ${\cal H}$}     
            \For {each $i \in \mathbb{M}(r)$}
                \State Let $\operatorname{pa}(r, i) := \{\widetilde{r} \in {\cal R}^+ \mid \widetilde{r} \rightarrow^{i} r \text{ in } {\cal H}\}$
                \State Fit $\widehat{p}(X_{i}^{(1)} \mid X_{-i}^{(1)}, O, R=r)$ using rows $\{D^n \mid r^n=\widetilde{r} \text{ for any } \widetilde{r} \in \operatorname{pa}(r, i)\}$
            \EndFor

            \State \Comment{Fill in missing values for all rows with pattern $r$ with Gibbs sampling using the fitted Gibbs factors}
            
            \For {$n=1, \ldots, N$}
                \If {$r^n = r$}
                    \State Gibbs sample $x^{(1)}_{\mathbb{M}(r)} \sim \widehat{p}(X^{(1)}_{\mathbb{M}(r)} \mid X_{\mathbb{O}(r)}=x^n_{\mathbb{O}(r)}, O=o^n, R=r)$\\ \hspace{1em} using Gibbs factors $\{\widehat{p}(X^{(1)}_{i} \mid X^{(1)}_{-i}, O, R=r) \mid \forall i \in \mathbb{M}(r)\}$
                    \State Update $\mathbb{D}$ by assigning $x^n_{\mathbb{M}(r)} \gets x^{(1)}_{\mathbb{M}(r)}$
                \EndIf
            \EndFor
        \EndFor
        \Return imputed data table $\mathbb{D}$
    \end{algorithmic}
\end{algorithm}

This forms the basis for MISPR (Multivariate Imputation via Supported Pattern Recursion) defined in
Algorithm~\ref{alg:imputation_via_pattern_dag}. The algorithm starts by determining identifiability of the full law based on the given m-graph $\G$ and the patterns with positive support in the data table $\mathbb{D}$. When  all patterns have positive support, we opt to use PM-ID as it provides some robustness to misspecification of the exact structure of the m-graph. Otherwise, we use PM-ID+.

If the identification step succeeds, we then process patterns according to any valid topological order of patterns $\prec_{\cal H}$, just as in Theorem~\ref{thm:pattern-dag-id}.
Suppose the algorithm is at some pattern $r$ with missing indices $\mathbb{M}(r)$. To estimate the $i^{th}$ Gibbs factor of the extrapolation density for $i\in \mathbb{M}(r)$, we find rows of data corresponding to patterns that we can borrow information from---this corresponds precisely to those patterns $\widetilde{r}$ such that $\widetilde{r} \rightarrow^i r$ exists in ${\cal H}$. By the induction hypothesis, all data rows corresponding to these patterns have been imputed in prior steps, and thus contain no missing values. Hence, the $i^{th}$ Gibbs factors can be fit using all data rows matching the patterns $\pa(r, i)$ computed in line 10 of the algorithm. Note these Gibbs factors can also be simplified according to the local Markov property, as $p(X_i^{(1)} \mid X_{-i}^{(1)}, O, R)=p(X_i^{(1)} \mid \mb_\G(X_i^{(1)})$, which makes estimation more feasible in high dimensional settings. Estimation can proceed either parametrically or using flexible machine learning models, such as random forests.\footnote{Though this can run the risk of leading to incongenial specifications of the joint extrapolation density.} The estimated Gibbs factors are then used to fill in the missing values for all data rows matching the current pattern $r$---this occurs in lines 13-16 of the algorithm. That is, in MISPR, imputation of data rows occurs according to a partial order on patterns with support, as outlined by the pattern DAG ${\cal H}$.

\section{Numerical Experiments}
\label{sec:experiments}

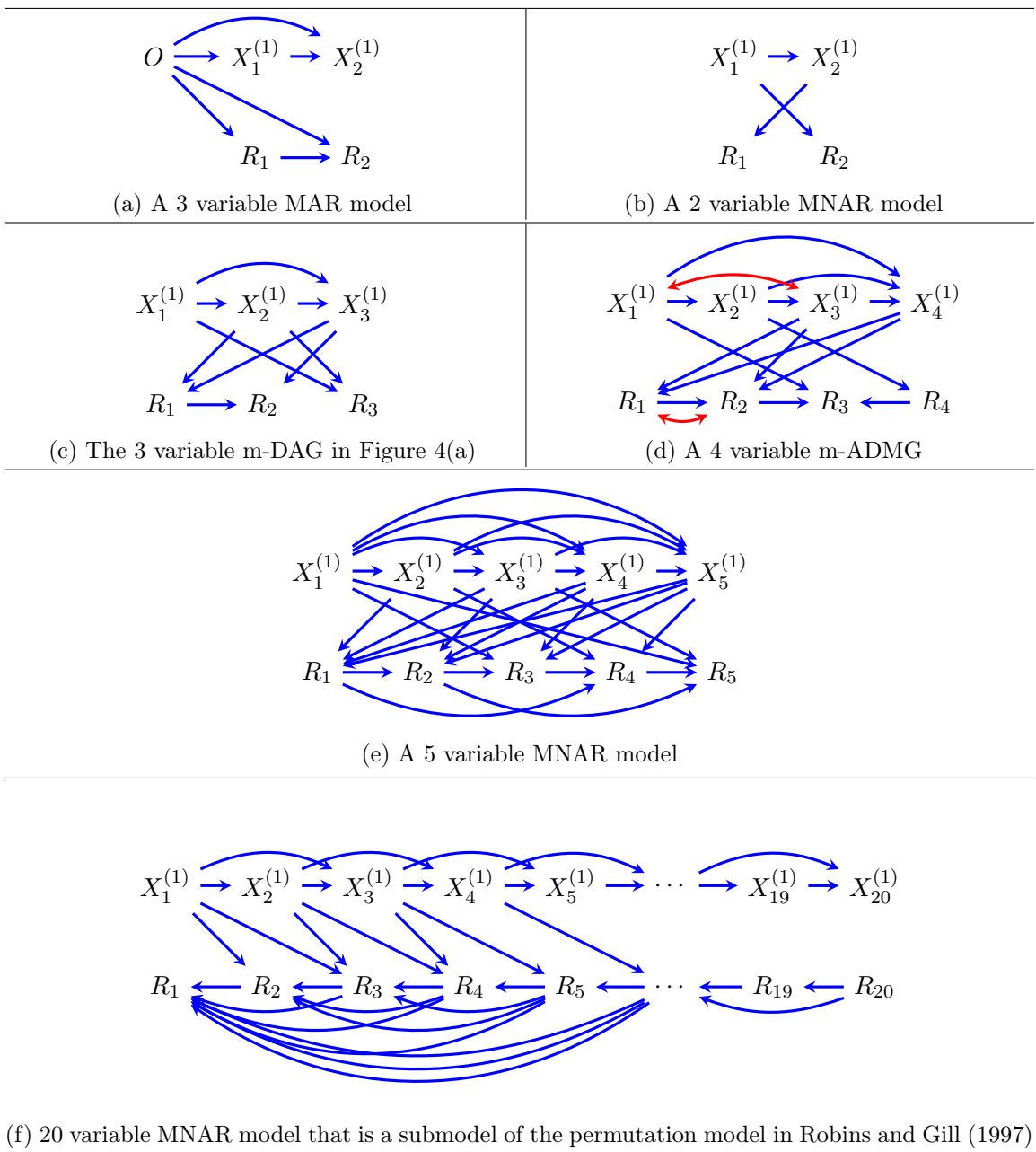
\begin{figure}
    \hrule
    \subcaptionbox{A 3 variable MAR model}[0.5\linewidth]{
        \begin{tikzpicture}[>=stealth, node distance=1.5cm]
            \def\d{1.5cm}
            \begin{scope}
                \path[->, very thick]
                node[] (o) {$O$}
                node[right of=o] (x11) {$X^{(1)}_1$}
                node[right of=x11] (x21) {$X^{(1)}_2$}
                node[below of=x11] (r1) {$R_1$}
                node[below of=x21] (r2) {$R_2$}
                
                (o) edge[blue] (x11)
                (x11) edge[blue] (x21)
                (o) edge[blue, bend left] (x21)
                (o) edge[blue] (r1)
                (o) edge[blue] (r2)
                (r1) edge[blue] (r2)
                ;
            \end{scope}
        \end{tikzpicture}
    }
    \vrule
    \subcaptionbox{A 2 variable MNAR model}[0.5\linewidth]{
        \begin{tikzpicture}[>=stealth, node distance=1.5cm]
            \def\d{1.5cm}
            \begin{scope}
                \path[->, very thick]
                node[] (x11) {$X^{(1)}_1$}
                node[right of=x11] (x21) {$X^{(1)}_2$}
                
                node[below of=x11] (r1) {$R_1$}
                node[below of=x21] (r2) {$R_2$}
                
                (x11) edge[blue] (r2)
                (x21) edge[blue] (r1)
                (x11) edge[blue] (x21)
                ;
            \end{scope}
        \end{tikzpicture}
    }
    \\
    \hrule
    \subcaptionbox{The 3 variable m-DAG in Figure~\ref{fig:pattern_dag_example_non_positive_1}(a)}[0.5\linewidth]{
        \begin{tikzpicture}[>=stealth, node distance=1.5cm]
            \def\d{1.5cm}
            \begin{scope}
                \path[->, very thick]
                node[] (x11) {$X^{(1)}_1$}
                node[right of=x11] (x21) {$X^{(1)}_2$}
                node[right of=x21] (x31) {$X^{(1)}_3$}
                
                node[below of=x11] (r1) {$R_1$}
                node[below of=x21] (r2) {$R_2$}
                node[below of=x31] (r3) {$R_3$}
                
                (r1) edge[blue] (r2)
                (x21) edge[blue] (r1)
                (x31) edge[blue] (r1)
                (x31) edge[blue] (r2)
                (x11) edge[blue] (r3)
                (x21) edge[blue] (r3)
                
                (x11) edge[blue] (x21)
                (x21) edge[blue] (x31)
                (x11) edge[blue, bend left] (x31)
                ;
            \end{scope}
        \end{tikzpicture}
    }
    \vrule
    \subcaptionbox{A 4 variable m-ADMG}[0.5\linewidth]{
        \begin{tikzpicture}[>=stealth, node distance=1.5cm]
            \def\d{1.5cm}
            \begin{scope}
                \path[->, very thick]
                node[] (x11) {$X^{(1)}_1$}
                node[right of=x11] (x21) {$X^{(1)}_2$}
                node[right of=x21] (x31) {$X^{(1)}_3$}
                node[right of=x31] (x41) {$X^{(1)}_4$}
                
                node[below of=x11] (r1) {$R_1$}
                node[below of=x21] (r2) {$R_2$}
                node[below of=x31] (r3) {$R_3$}
                node[below of=x41] (r4) {$R_4$}
                
                (r1) edge[blue] (r2)
                (r2) edge[blue] (r3)
                (r4) edge[blue] (r3)
                (x21) edge[blue] (r4)
                (x31) edge[blue] (r1)
                (x31) edge[blue] (r2)
                (x11) edge[blue] (r3)
                (x41) edge[blue] (r1)
                (x41) edge[blue] (r2)
                (r1) edge[red, <->, bend right=25] (r2)
                
                (x11) edge[blue] (x21)
                (x21) edge[blue] (x31)
                (x31) edge[blue] (x41)
                (x21) edge[blue, bend left=20] (x41)
                (x11) edge[blue, bend left = 35] (x41)
                (x11) edge[red, <->, bend left=20] (x31)
                ;
            \end{scope}
        \end{tikzpicture}
    }
    \\
    \hrule
    \subcaptionbox{A 5 variable MNAR model}[1.0\linewidth]{
        \centering
        \begin{tikzpicture}[>=stealth, node distance=1.5cm]
            \def\d{1.5cm}
            \begin{scope}
                \path[->, very thick]
                node[] (x11) {$X^{(1)}_1$}
                node[right of=x11] (x21) {$X^{(1)}_2$}
                node[right of=x21] (x31) {$X^{(1)}_3$}
                node[right of=x31] (x41) {$X^{(1)}_4$}
                node[right of=x41] (x51) {$X^{(1)}_5$}
                
                node[below of=x11] (r1) {$R_1$}
                node[below of=x21] (r2) {$R_2$}
                node[below of=x31] (r3) {$R_3$}
                node[below of=x41] (r4) {$R_4$}
                node[below of=x51] (r5) {$R_5$}
                
                (r1) edge[blue] (r2)
                (r1) edge[blue, bend right=25] (r4)
                (r2) edge[blue] (r3)
                (r2) edge[blue, bend right=25] (r5)
                (r3) edge[blue] (r4)
                (r4) edge[blue] (r5)
                
                (x11) edge[blue] (r3)
                (x11) edge[blue] (r5)
                (x21) edge[blue] (r1)
                (x21) edge[blue] (r4)
                (x31) edge[blue] (r1)
                (x31) edge[blue] (r2)
                (x31) edge[blue] (r5)
                (x41) edge[blue] (r1)
                (x41) edge[blue] (r2)
                (x41) edge[blue] (r3)
                (x51) edge[blue] (r1)
                (x51) edge[blue] (r2)
                (x51) edge[blue] (r3)
                (x51) edge[blue] (r4)
                
                (x11) edge[blue,] (x21)
                (x11) edge[blue, bend left=25] (x31)
                (x11) edge[blue, bend left=30] (x41)
                (x11) edge[blue, bend left=35] (x51)
                (x21) edge[blue,] (x31)
                (x21) edge[blue, bend left=25] (x41)
                (x21) edge[blue, bend left] (x51)
                (x31) edge[blue,] (x41)
                (x31) edge[blue, bend left=25] (x51)
                (x41) edge[blue,] (x51)
                ;
                \end{scope}
        \end{tikzpicture}
    }
    \\
    \vspace{0.2em}
    \hrule
    \vspace{1cm}
    \subcaptionbox{20 variable MNAR model that is a submodel of the permutation model in \cite{robins1997non}}[1.0\linewidth]{
        \centering
        \begin{tikzpicture}[>=stealth, node distance=1.5cm]
            \def\d{1.5cm}
            \begin{scope}
                \path[->, very thick]
                node[] (x11) {$X^{(1)}_1$}
                node[right of=x11] (x21) {$X^{(1)}_2$}
                node[right of=x21] (x31) {$X^{(1)}_3$}
                node[right of=x31] (x41) {$X^{(1)}_4$}
                node[right of=x41] (x51) {$X^{(1)}_5$}
                node[right of=x51] (xdd) {$\cdots$}
                node[right of=xdd] (x181) {$X^{(1)}_{19}$}
                node[right of=x181] (x191) {$X^{(1)}_{20}$}
                
                node[below of=x11] (r1) {$R_1$}
                node[below of=x21] (r2) {$R_2$}
                node[below of=x31] (r3) {$R_3$}
                node[below of=x41] (r4) {$R_4$}
                node[below of=x51] (r5) {$R_5$}
                node[below of=xdd] (rdd) {$\cdots$}
                node[below of=x181] (r18) {$R_{19}$}
                node[below of=x191] (r19) {$R_{20}$}
                
                (r2) edge[blue] (r1)
                (r3) edge[blue, bend left=20] (r1)
                (r4) edge[blue, bend left=25] (r1)
                (r5) edge[blue, bend left=30] (r1)
                (rdd) edge[blue, bend left=35] (r1)
                (r3) edge[blue] (r2)
                (r4) edge[blue, bend left=20] (r2)
                (r5) edge[blue, bend left=25] (r2)
                (rdd) edge[blue, bend left=30] (r1)
                (r4) edge[blue] (r3)
                (r5) edge[blue, bend left=20] (r3)
                (rdd) edge[blue, bend left=25] (r1)
                (r5) edge[blue] (r4)
                (rdd) edge[blue] (r5)
                (r18) edge[blue] (rdd)
                (r19) edge[blue, bend left=20] (rdd)
                (r19) edge[blue] (r18)
                
                (x11) edge[blue] (r2)
                (x11) edge[blue] (r3)
                (x21) edge[blue] (r3)
                (x21) edge[blue] (r4)
                (x31) edge[blue] (r4)
                (x31) edge[blue] (r5)
                (x41) edge[blue] (rdd)
                
                (x11) edge[blue,] (x21)
                (x11) edge[blue, bend left=25] (x31)
                (x21) edge[blue,] (x31)
                (x21) edge[blue, bend left=25] (x41)
                (x31) edge[blue,] (x41)
                (x31) edge[blue, bend left=25] (x51)
                (x41) edge[blue,] (x51)
                (x51) edge[blue,] (xdd)
                (x41) edge[blue, bend left=25] (xdd)
                (xdd) edge[blue,] (x181)
                (x181) edge[blue,] (x191)
                (xdd) edge[blue, bend left=25] (x191)
                ;
                \end{scope}
        \end{tikzpicture}
    }
    \vspace{0.2em}
    \hrule
    \caption{Missing data models used in our numerical experiments.}
    \label{fig:m_graphs_experiments}
\end{figure}

In this section we evaluate the performance of MISPR in comparison with the MICE package available in the \texttt{R} programming language \citep{van2011mice}, as it is currently the most popular off-the-shelf imputation software. To compare the asymptotic behavior of these methods, for each synthetic data generating process (DGP), we generate $N=10^5$ samples for the smaller examples in Figures~\ref{fig:m_graphs_experiments}(a-e) and generate $N=5\times 10^5$ for the high-dimensional example in Figure~\ref{fig:m_graphs_experiments}(f). Details of each DGP are provided in Appendix~\ref{app:dgps}.

For each experiment, we generate an observed dataset $\mathbb{D}$ with $N$ samples by sampling according to the true full law $p_0(X^{(1)}, O, R)$, i.e., we draw $N$ samples of $X^{(1)}, O, R \sim p_0$, and then set missing variables to be unobserved according to the rules of missing data consistency. The true full law $p_0$ for each experiment is selected randomly, so as not to pick any specific DGPs that may be favorable to one method or the other. Then, imputation is performed separately on the dataset with missing values $\mathbb{D}$ using MICE and MISPR with PM-ID+ and PM-ID in the identification step, when applicable. We use MICE with random forests to estimate each imputing distribution using the settings \texttt{mice(data, m = 7, method = `rf', seed = 123, maxit=100)}. For MISPR, we also fit the Gibbs factors using random forests, and use a burn-in period of 500 steps for Gibbs sampling. We then use the aggregated datasets to obtain estimates of the target law $\widehat{p}(X^{(1)}, O)$ using maximum likelihood. We obtain $7$ imputed datasets $\widehat{\mathbb{D}}_1, \dots, \widehat{\mathbb{D}}_7$ from each method using different random seeds in this manner, and aggregate them as $\widehat{\mathbb{D}} = \bigcup_{i=1}^7 \widehat{\mathbb{D}}_i$.  We then compare the accuracy of imputation using MISPR and MICE by computing the $L_2$ and $L_\infty$ distance between the fitted and ground truth target laws, i.e., $\|\hat{p}(X^{(1)}, O) - p_0(X^{(1)}, O)\|_{q}$ with $q=2$ and $q=\infty$. We also report timings for each method in producing the seven imputed datasets. Table~\ref{tab:mice_pm_compare} reports these distances and timings for both MISPR and MICE across all DGPs. Entries marked with a `+' indicate the DGP contains a positivity violation. We briefly summarize the main findings below, where each bold subheading highlights the main finding from the experiment.

\paragraph{Expr 1: Ignoring positivity violations can lead to bias, even in MAR scenarios} In this experiment, data are generated from two  full data laws that factorize according to the m-graph in Figure~\ref{fig:m_graphs_experiments}(a), where one $X^{(1)} \Perp R \mid O$ so that the model is MAR. The first full data law we sample from is strictly positive, but the second is not. In the former case, MISPR (using either PM-ID+ or PM-ID) and MICE achieve fairly comparable results, albeit MICE is much slower and the bias from MISPR is also an order of magnitude lower (first row of Table~\ref{tab:mice_pm_compare}). In the second case, however, the difference is much more noticeable, where despite the structural assumptions of the m-graph implying a MAR model, the correct handling of patterns with no support using PM-ID+ is essential. The second row of Table~\ref{tab:mice_pm_compare} shows how MISPR with PM-ID+ obtains much lower bias than MICE in this setting.

\paragraph{Expr 2: MISPR is able to handle both MNAR and sparse pattern support}
This experiment checks the performance of MISPR across a variety of MNAR full data laws, with and without positivity violations. This corresponds to full data laws drawn from Figures~\ref{fig:m_graphs_experiments}(b-d), with results reported in the corresponding rows of Table~\ref{tab:mice_pm_compare}. MISPR outperforms MICE across the board in timing and bias, with the exception of one scenario where MICE obtains comparable results. This corresponds to a full data law with positivity violations that factorizes according to Figure~\ref{fig:m_graphs_experiments}(b). On inspection of this DGP, we realized that the particular positivity violations and the structural assumptions of this model actually reduce to a MAR model, which is an exceptionally rare occurrence in general.

\paragraph{Expr 3: Imputation with several missing variables and sparse pattern support}
In this experiment, we generate data from full data laws that factorize according to the m-graphs in Figures~\ref{fig:m_graphs_experiments}(e) and (f) containing 5 and 20 variables respectively. In such high-dimensions support on all possible patterns is exceedingly unlikely---in our DGPs the full data law for Figure~\ref{fig:m_graphs_experiments}(e) has $24$ out of the $2^5$ possible patterns, while the one for Figure~\ref{fig:m_graphs_experiments}(e) has only $476$ out of the possible $2^{20}$ possible patterns (explicit details of these DGPs can be found in the Appendix~\ref{app:dgps}). For the 5 variable case we are able to produce comparisons to MICE that demonstrate that MISPR with PM-ID+ produces superior results. For the 20 variable case, the scaling of the MICE package did not allow us to obtain results for it---MICE took close to a day even for the 3-5 variable examples. Thus, we only run MISPR on the 20 variable example, and show that it produces estimates with very low bias, with reasonable running time on the 476 patterns. We suspect these running times can be improved even further with a better implementation of our method---our implementation is in Python, and does not use any parallelism to process patterns in the partial order that are incompatible (two patterns $r$ and $\widetilde{r}$ are considered incompatible if neither $r \prec_{\cal H} \widetilde{r}$ nor $\widetilde{r} \prec_{\cal H} r$).

\paragraph{Summary} MISPR outperforms MICE across a wide range of scenarios, including MAR models with positivity violations and MNAR models with and without positivity violations.

\def\arraystretch{1.5}
\begin{table}[t]
\begin{tabular}
{||p{2.5em}||p{2.5em}|p{2.5em}|p{3.5em}||p{2.5em}|p{2.5em}|p{3.5em}||p{2.5em}|p{2.5em}|p{3.5em}||}
    \hline
    \multirow{2}{4em}{DGP} & \multicolumn{3}{|c||}{\bf MICE} & \multicolumn{3}{|c||}{{\bf MISPR w/ PM-ID+}} & \multicolumn{3}{|c||}{{\bf MISPR w/ PM-ID}} \\
    \cline{2-10}
    & $\text{L}_{2}$ & $\text{L}_{\infty}$ & Time & $\text{L}_{2}$ & $\text{L}_{\infty}$ & Time & $\text{L}_{2}$ & $\text{L}_{\infty}$ & Time \\
    
    \hline
    \ref{fig:m_graphs_experiments}a
        & 0.0890 & 0.0568 & 6h26m
        & 0.0060 & 0.0034 & 1h13m
        & 0.0031 & 0.0020 & 1h18m
    \\
    \ref{fig:m_graphs_experiments}a +
        & 0.2329 & 0.1576 & 5h24m
        & 0.0030 & 0.0016 & 0h29m
        & NA & NA & NA
    \\
    
    \ref{fig:m_graphs_experiments}b
        & 0.2580 & 0.1913 & 3h16m
        & 0.0039 & 0.0030 & 0h50m
        & 0.0038 & 0.0028 & 0h50m
    \\
    \ref{fig:m_graphs_experiments}b +
        & 0.0028 & 0.0023 & 9h17m
        & 0.0026 & 0.0021 & 11 sec
        & NA & NA & NA
    \\
    
    
    \ref{fig:m_graphs_experiments}c
        & 0.0854 & 0.0568 & 18h
        & 0.0075 & 0.0060 & 1h04m
        & 0.0078 & 0.0061 & 1h23m
    \\
    \ref{fig:m_graphs_experiments}c +
        & 0.0197 & 0.0138 & 23h36m
        & 0.0045 & 0.0034 & 0h44m
        & NA & NA & NA
    \\

    \ref{fig:m_graphs_experiments}d
        & 0.0678 & 0.0457 & 30h47m
        & 0.0074 & 0.0034 & 1h28m
        & 0.0057 & 0.0029 & 1h30m
    \\
    \ref{fig:m_graphs_experiments}d +
        & 0.0575 & 0.0406 & 42h53m
        & 0.0099 & 0.0058 & 1h41m
        & NA & NA & NA
    \\

    \ref{fig:m_graphs_experiments}e +
        & 0.1205 & 0.0630 & 18h30m
        & 0.0412 & 0.0324 & 3h
        & NA & NA & NA
    \\

    \ref{fig:m_graphs_experiments}f +
        & - & - & -
        & $2.4 \times 10^{-5}$ & $3.3 \times 10^{-7}$ & 28h38m
        & NA & NA & NA
    \\
    
    \hline
\end{tabular}
\caption{
    Comparison between MICE and MISPR on random full data laws that factorize according to m-graphs in Figure~\ref{fig:m_graphs_experiments}. Rows marked with `+' correspond to full data laws that have positivity violations. The $L_2$ and $L_{\infty}$ distance of the true and estimated target laws are reported to measure the accuracy of the imputation. PM-ID is not applicable when there are positivity violations, hence we report ``NA'' for the corresponding rows. We were unable to obtain results for MICE for the full data law \ref{fig:m_graphs_experiments}f + due to its long run times.
}
\label{tab:mice_pm_compare}
\end{table}



\section{Discussion And Related Work}
\label{sec:comparison}

In this paper, we have added to the identification and imputation literature on missing data models, specifically those that factorize according to missing data graphs.

From an identification standpoint, our work is most similar to the identification methods proposed by \cite{sadinle2016itemwise} and the pattern graph framework proposed by \cite{Chen.2022.PatternGraphs}. \cite{sadinle2016itemwise} provide a constructive argument for full law identification of the no self-censoring model (recall, a 3 variable version of this model is shown in Figure~\ref{fig:m-graphs}(c)) based on a log-linear parameterization of the full law. However, this constructive argument fails when there is positive support for only a subset of all missingness patterns, which we deal with in Section~\ref{sec:positivity-violations}. \cite{Chen.2022.PatternGraphs} proposes a pattern graph representation of assumptions in missing data models, where the patterns with positive support are represented as vertices on a graph. The author then provides a factorization of the model with respect to the pattern graph itself. In contrast, the pattern graphs we propose here are not used to define the model, but rather encode identification assumptions implied by the factorization of an associated m-graph. That is, our use of pattern graphs still allows the analyst to specify assumptions in an intuitive manner based on missing edges in an m-graph. Another key difference is the restriction on edges in the two types of pattern graphs---in \cite{Chen.2022.PatternGraphs}, if  $\widetilde{r} \rightarrow r$ exists in ${\cal H}$, then it must be the case that $\mathbb{M}(r) \supset \mathbb{M}(\widetilde{r})$. This is not the case for us, as demonstrated in Figure~\ref{fig:pattern_dag_example_non_positive_2}(b), where the edge $(001)\rightarrow^3(100)$ does not follow this rule.


From an imputation standpoint, our work is most similar to \cite{kyono2021miracle}, \cite{ren2023multiple}, and \cite{karvanen2024multiple}. In \cite{kyono2021miracle} and \cite{karvanen2024multiple}, the authors propose imputation methods for missing data models whose full laws are identified and can be represented as m-DAGs. Besides covering additional classes of m-graphs, our method is more general, even for m-DAGs. In \cite{kyono2021miracle}, the missingness indicators in the m-DAG are not allowed to have any outgoing edges, while \cite{karvanen2024multiple} focus on m-DAGs that can be identified via a sequential procedure, i.e., models that can be identified via a total ordering. However, as we have seen in Section~\ref{sec:constructive_proof}, partial orders play an important role for a sound and complete procedure, even for m-DAGs. \cite{ren2023multiple} propose a multiple imputation procedure for the no self-censoring model under certain parametric assumptions. Moreover, the identification step of their procedure requires full support on all possible missingness patterns, which we relax with PM-ID+ by leveraging other constraints that might be implied by the given m-graph. The authors do, however, provide more discussion on congenial specification of the imputing distributions in the no self-censoring model and ideas for sensitivity analysis to departures from the structural assumptions of the model. These contributions are complementary to our own work in this paper.

Finally we conclude with some brief remarks on how our algorithm differs from MICE \citep{raghunathan2001multivariate, van2007multiple}. As we have seen in Section~\ref{sec:experiments}, MISPR outperforms MICE when the data are MNAR. The main reason for this is that MICE (and other similar multiple imputation algorithms) operates under a Missing At Random (MAR) assumption. As mentioned earlier, MICE also employs a Gibbs sampling algorithm when filling in missing values. Due to the MAR assumption, however, rather than borrowing information from different patterns according to a set of constraints encoded by a pattern DAG, MICE allows any pattern to borrow information from any other pattern when fitting the Gibbs factors. This would be valid for MAR models, but can fail for even simple MNAR models, including some MNAR models where complete case analysis yields unbiased estimates of target parameters \citep{mathur2025pitfalls}.\footnote{See also \cite{mathur2024imputation} for graphical criteria for m-DAGs and m-ADMGs describing when MICE can work without modification.}

\section{Conclusion}
\label{sec:conclusion}

We have presented two constructive algorithms---PM-ID and PM-ID+---for identification of the full law in graphical models of missing data. PM-ID+ is able to identify full laws even when certain patterns of missingness have no support. An interesting open question is related to the completeness of PM-ID+: Is it complete for unrestricted target laws, and if not, what modifications are required for it be complete. Any improvements to these identification algorithms, including maximizing data efficiency when connecting patterns in the pattern DAG, would lead to immediate improvements for our imputation algorithm MISPR. Another interesting avenue for future research is to fully characterize the intersection between the pattern graph framework of \cite{Chen.2022.PatternGraphs} and the pattern graphs used in PM-ID and PM-ID+. This would lead to more intuitive representations of a wide class of pattern graphs as equivalent m-graphs.

\subsection*{Acknowledgements}
\label{sec:acks}

Rohit Bhattacharya acknowledges support from the NSF CRII grant 234828. Ilya Shpitser acknowledges support from grants ONR N000142412701, NIH R01HS027819, NSF CAREER 1942239. The content of the information does not necessarily reflect the position or the policy of the Government, and no official endorsement should be inferred.

\clearpage

\bibliography{references}

\newpage
\appendix

\section{PM-ID When No Complete Cases Are Available}
\label{app:no_complete_cases}

\begin{figure}[H]
    \centering

    \subcaptionbox{}{
    \begin{tikzpicture}[>=stealth, node distance=1.5cm]
        \def\d{1.5cm}
        \tikzstyle{hid} = [draw, circle, red, inner sep=0]
        \begin{scope}
            \path[->, very thick]
            node[] (X1) {$X^{(1)}_1$}
            node[right of=X1] (X2) {$X^{(1)}_2$}
            node[right of=X2] (X3) {$X^{(1)}_3$}
            node[below of=X1] (R1) {$R_1$}
            node[below of=X2] (R2) {$R_2$}
            node[below of=X3] (R3) {$R_3$}
            
            (X2) edge[blue] (X3)
            (X1) edge[blue] (R3)
            (X3) edge[blue] (R1)
            (X3) edge[blue] (R2)
            
            (R1) edge[blue] (R2)
            ;
        \end{scope}
    \end{tikzpicture}
    }
    \hspace{0.05\linewidth}
    \subcaptionbox{}{
    \scalebox{0.75}{
    \begin{tikzpicture}
            \def\d{2.0cm}
            \begin{scope}[>=stealth, node distance=2.0cm]
                \path[->, very thick]
                node[] (111) {${(111)}$}
                node[below of=111] (101) {$(101)$}
                node[right of=101] (011) {$(011)$}
                node[left of=101] (110) {${(110)}$}
                node[right of=111] (1x1) {$(1 * 1)$}
                node[left of=111] (x11) {$(* 1 1)$}
                
                node[below of=110] (100) {$(100)$}
                node[below of=101] (010) {$(010)$}
                node[below of=011] (001) {$(001)$}
                node[below of=010] (000) {$(000)$}
                
                (1x1) edge[blue] node [midway, right]{1} (011)
                (x11) edge[blue] node [midway, left]{2} (101)
                (1x1) edge[blue, bend left] node [midway, right]{1} (001)
                (x11) edge[blue, bend left] node [midway, right]{2} (001)
                (x11) edge[blue, bend right] node [midway, left]{2} (100)
                (101) edge[blue] node [midway, left]{3} (100)
                
                (100) edge[blue] node [midway, above]{1} (010)
                (011) edge[blue] node [midway, left]{3} (010)
                  
                (100) edge[blue] node [midway, left]{1} (000)
                (010) edge[blue] node [midway, left]{2} (000)
                (001) edge[blue] node [midway, right]{3} (000)
                ;
                \draw[red, very thick] ([xshift=-2pt,yshift=0pt]111.west) -- ([xshift=2pt,yshift=0pt]111.east);
                \draw[red, very thick] ([xshift=-2pt,yshift=0pt]110.west) -- ([xshift=2pt,yshift=0pt]110.east);
            \end{scope}
        \end{tikzpicture}}
        }

    \caption{
        (a) An m-graph with a sparse structure for variables in $X^{(1)}$  used to demonstrate how identification can be achieved when there are no complete cases. For this example, we assume the patterns $\{ 111, 110 \}$ have no support. (b) Pattern DAG used to identify the full law without using the complete case pattern. The marginals $p(X^{(1)}_{23}, R_{2} = 1, R_{3}=1)$ and $p(X^{(1)}_{13}, R_{1} = 1, R_{3}=1)$ are denoted as $(1*1)$ and $(*11)$, respectively.
    }
    \label{fig:no_complete_case}
\end{figure}
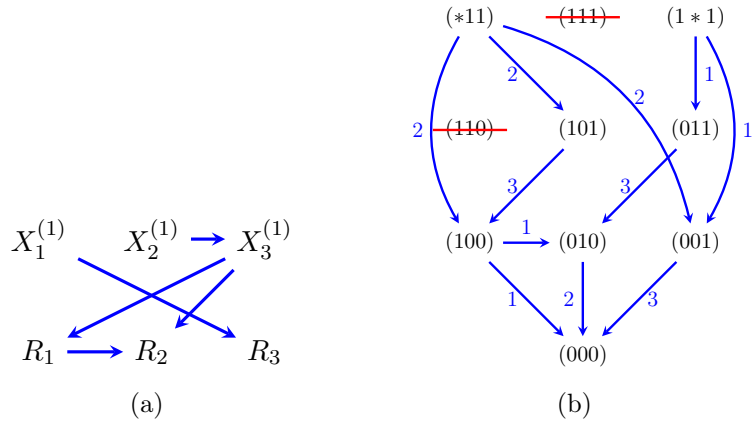

Consider a missing data model that factorizes according to the m-graph in Figure~\ref{fig:no_complete_case}(a) and has a positivity violation on one of the propensity scores: $p(R_{2} = 1 \mid R_{1}=1, X^{(1)}_{3}) = 0$. This gives us a model that has no support for the complete case pattern $111$ as well as the pattern $110$. Thus, ${\cal R}^+ = \{ 011, 101, 001, 010, 100, 000 \}$. Recall that PM-ID and PM-ID+ require the complete case pattern to be the root node of the pattern DAG, and serve as the base case for the identification  argument in Theorem~\ref{thm:pattern-dag-id}. Though the complete case pattern has no support here, we demonstrate that the recursive argument can still be applied by expanding the set of patterns considered to also include identified marginals of the full law.

In our problem, two marginals that are identified at the outset are,
\begin{align}
\label{eq:marginals}
& p(X^{(1)}_{13}, R_{1} = 1, R_{3}=1) \text{ and } p(X^{(1)}_{23}, R_{2} = 1, R_{3}=1).
\end{align}
The marginal $p(X^{(1)}_{12}, R_{1} = 1, R_{2}=1)$ involves a positivity violation as a direct consequence of there being no support on the patterns $\{ 111, 110 \}$, and thus cannot be used. Nonetheless, we can extend the the recursive arguments of PM-ID and PM-ID+ by treating the identified marginals in \eqref{eq:marginals} as two distinct root nodes of the pattern DAG.

Consider the pattern DAG in Figure~\ref{fig:no_complete_case}(b), where the two marginals in \eqref{eq:marginals} are set up as root nodes (the notation *  is used to indicate variables that are marginalized.) Now, to achieve identification of other patterns, we use conditional independences in the target law implied by the m-graph shown in Figure~\ref{fig:no_complete_case}(a). Recall that PM-ID used independences of the form $X_i^{(1)} \Perp R_i \mid \mb_\G(X_i^{(1)})$, while PM-ID+ extended this to also use independences of the form $X_i^{(1)} \Perp R_j \mid \mb_\G(X_i^{(1)})$ for $i \not= j$. An extension of these when there are no complete cases must also use independences of the form $X_i^{(1)} \Perp X_j^{(1)} \mid \mb_\G(X_i^{(1)})$ for $i \not= j$ that arise due to missing edges among variables in $X^{(1)}$.

We identify the Gibbs factors of the extrapolation densities of the rest of the patterns by arguments encoded by the pattern DAG shown in Figure~\ref{fig:no_complete_case}(b).
Notably this pattern DAG uses the marginal $p(X^{(1)}_{23}, R_{2} = 1, R_{3}=1)$ to identify Gibbs factors for the patterns $101, 100, 001$, and uses the marginal $p(X^{(1)}_{13}, R_{1} = 1, R_{3}=1)$ to identify the Gibbs factors for the pattern $011, 001$.
The equations in \eqref{eq:pm-id*} make the arguments encoded in the pattern DAG more explicit, and also list the independences that lead to the equalities that permit identification. Such an algorithm might be labeled PM-ID*, by analogy with PM-ID+ due to its usage of margins of the full data law that have support as root nodes.\footnote{In theory, PM-ID* exhausts the capability of the pattern mixture method employed in this paper. The other independence statements we have not used is the local Markov property relative to indicators $R$, $R_{i} \Perp V \setminus \mb_{\mathcal{G}}(R_{i}) \mid \mb_{\mathcal{G}}(R_{i})$. We expect that a sound and complete identification theory for the class of models we explore under positivity violation will make use of both $R_{i}$ and $X^{(1)}_i$ local Markov property.}

{\small
\begin{equation}
\label{eq:pm-id*}
\begin{aligned}
(011): &
\\
& & p(X^{(1)}_{1} \mid X^{(1)}_{23}, R=011)
    &=
    p(X^{(1)}_{1} \mid R_{1}=1, R_{3}=1)
    &
    (X^{(1)}_{1} \Perp R_{12}, X^{(1)}_{23} \mid R_{3})
    \\
(101): &
\\
& & p(X^{(1)}_{2} \mid X^{(1)}_{13}, R=101)
    &=
    p(X^{(1)}_{2} \mid X^{(1)}_{3}, R_{2}=1, R_{3}=1)
    &
    (X^{(1)}_{2} \Perp R, X^{(1)}_{1} \mid X^{(1)}_{3})
    \\
(001): &
\\
& & p(X^{(1)}_{1} \mid X^{(1)}_{23}, R=001)
    &=
    p(X^{(1)}_{1} \mid R_{1}=1, R_{3}=1)
    &
    (X^{(1)}_{1} \Perp R_{12}, X^{(1)}_{23} \mid R_{3})
    \\
& & p(X^{(1)}_{2} \mid X^{(1)}_{13}, R=001)
    &=
    p(X^{(1)}_{2} \mid X^{(1)}_{3}, R_{2}=1, R_{3}=1)
    &
    (X^{(1)}_{2} \Perp R, X^{(1)}_{1} \mid X^{(1)}_{3})
    \\
(100): &
\\
& & p(X^{(1)}_{2} \mid X^{(1)}_{13}, R=100)
    &=
    p(X^{(1)}_{2} \mid X^{(1)}_{3}, R_{2}=1, R_{3}=1)
    &
    (X^{(1)}_{2} \Perp R, X^{(1)}_{1} \mid X^{(1)}_{3})
    \\
& & p(X^{(1)}_{3} \mid X^{(1)}_{12}, R=100)
    &=
    p(X^{(1)}_{3} \mid X^{(1)}_{12}, R=101)
    &
    (X^{(1)}_{3} \Perp R_{3} \mid X^{(1)}_{12}, R_{12})
    \\
(010): &
\\
& & p(X^{(1)}_{1} \mid X^{(1)}_{23}, R=010)
    &=
    p(X^{(1)}_{1} \mid X^{(1)}_{23}, R=100)
    &
    (X^{(1)}_{1} \Perp R_{12} \mid X^{(1)}_{23}, R_{3})
    \\
& & p(X^{(1)}_{3} \mid X^{(1)}_{12}, R=010)
    &=
    p(X^{(1)}_{3} \mid X^{(1)}_{12}, R=011)
    &
    (X^{(1)}_{3} \Perp R_{3} \mid X^{(1)}_{12}, R_{12})
    \\
(000): &
\\
& & p(X^{(1)}_{1} \mid X^{(1)}_{23}, R=000)
    &=
    p(X^{(1)}_{1} \mid X^{(1)}_{23}, R=100)
    &
    (X^{(1)}_{1} \Perp R_{1} \mid X^{(1)}_{23}, R_{23})
    \\
& & p(X^{(1)}_{2} \mid X^{(1)}_{13}, R=000)
    &=
    p(X^{(1)}_{2} \mid X^{(1)}_{13}, R=010)
    &
    (X^{(1)}_{2} \Perp R_{2} \mid X^{(1)}_{13}, R_{13})
    \\
& & p(X^{(1)}_{3} \mid X^{(1)}_{12}, R=000)
    &=
    p(X^{(1)}_{3} \mid X^{(1)}_{12}, R=001)
    &
    (X^{(1)}_{3} \Perp R_{3} \mid X^{(1)}_{12}, R_{12})
    \\
\end{aligned}
\end{equation}}

In conclusion, the full law is identified.

\newpage

\section{Data Generating Processes For Experiments}
\label{app:dgps}

In our experiment, all variables are binary. The following tables give the values of $p(V=0 \mid W)$, where $V=0$ is the name of the last column, and $W$ are all the variables whose values are given in the other columns.

\paragraph{DGPs for graph (a)}

\begin{itemize}
    \item With positivity
\end{itemize}

\begin{equation*}
\tiny
\begin{array}{r}
    p(O = 0) = 0.660453
    \\
    \\
    \begin{array}{| c | l |}
    \hline
        &     X1=0 \\
    \hline
    O   &          \\
    0   & 0.786893 \\
    1   & 0.637722 \\
    \hline
    \end{array}
\end{array}
\quad \quad
\begin{array}{| c | c | l |}
\hline
   &    &     X2=0 \\
\hline
X1 & O  &          \\
0  & 0  & 0.513861 \\
   & 1  & 0.429143 \\
1  & 0  & 0.664446 \\
   & 1  & 0.149863 \\
\hline
\end{array}
\quad \quad
\begin{array}{| c | l |}
\hline
    &     R1=0 \\
\hline
O   &          \\
0   & 0.883374 \\
1   & 0.481486 \\
\hline
\end{array}
\quad \quad
\begin{array}{| c | c | l |}
\hline
   &    &     R2=0 \\
\hline
O  & R1 &          \\
0  & 0  & 0.891517 \\
   & 1  & 0.310271 \\
1  & 0  & 0.530085 \\
   & 1  & 0.470104 \\
\hline
\end{array}
\end{equation*}

\begin{itemize}
    \item With positivity violation
\end{itemize}

\begin{equation*}
\tiny
\begin{array}{r}
    p(O = 0) = 0.512816
    \\
    \\
    \begin{array}{| c | l |}
    \hline
        &     X1=0 \\
    \hline
    O   &          \\
    0   & 0.883374 \\
    1   & 0.481486 \\
    \hline
    \end{array}
\end{array}
\quad \quad
\begin{array}{| c | c | l |}
\hline
   &    &     X2=0 \\
\hline
X1 & O  &          \\
0  & 0  & 0.891517 \\
   & 1  & 0.310271 \\
1  & 0  & 0.530085 \\
   & 1  & 0.470104 \\
\hline
\end{array}
\quad \quad
\begin{array}{| c | l |}
\hline
    &     R1=0 \\
\hline
O   &          \\
0   & 0.970209 \\
1   & 0.049576 \\
\hline
\end{array}
\quad \quad
\begin{array}{| c | c | l |}
\hline
   &    &     R2=0 \\
\hline
O  & R1 &          \\
0  & 0  & 0.430232 \\
   & 1  & 0        \\
1  & 0  & 0.713360 \\
   & 1  & 0        \\
\hline
\end{array}
\end{equation*}

\paragraph{DGPs for graph (b)}

\begin{itemize}
    \item With positivity
\end{itemize}

\begin{equation*}
\tiny
\begin{array}{r}
    p(X1=0) = 0.660453
    \\
    \begin{array}{| c | l |}
    \hline
       & X2=0     \\
    \hline
    X1 &          \\
    0  & 0.786893 \\
    1  & 0.637722 \\
    \hline
    \end{array}
\end{array}
\quad \quad 
\begin{array}{| c | l |}
\hline
   & R1=0     \\
\hline
X2 &          \\
0  & 0.436589 \\
1  & 0.920788 \\
\hline
\end{array}
\quad \quad 
\begin{array}{| c | l |}
\hline
   & R2=0     \\
\hline
X1 &          \\
0  & 0.883374 \\
1  & 0.481486 \\
\hline
\end{array}
\end{equation*}

\begin{itemize}
    \item With positivity violation
\end{itemize}

\begin{equation*}
\tiny
\begin{array}{r}
    p(X1=0) = 0.512816
    \\
    \begin{array}{| c | l |}
    \hline
        & X2=0     \\
    \hline
    X1  &          \\
    0   & 0.883374 \\
    1   & 0.481486 \\
    \hline
    \end{array}
\end{array}
\quad \quad 
\begin{array}{| c | l |}
\hline
    & R1=0     \\
\hline
X2  &          \\
0   & 0.474077 \\
1   & 0.386253 \\
\hline
\end{array}
\quad \quad 
\begin{array}{| c | l |}
\hline
    & R2=0     \\
\hline
X1  &          \\
0   & 0.0      \\
1   & 0.0      \\
\hline
\end{array}
\end{equation*}

\newpage

\paragraph{DGPs for graph (c)}

\begin{itemize}
    \item With positivity
\end{itemize}

\begin{equation*}
\tiny
\begin{array}{r}
    p(X^{(1)}_1 = 0) = 0.512816
    \\
    \\
    \begin{array}{| c | l |}
        \hline
           &  X2=0     \\
        \hline
        X1 &           \\
        0  &  0.883374 \\
        1  &  0.481486 \\
        \hline
    \end{array}
\end{array}
\quad \quad
\begin{array}{| c | c | l |}
\hline
   &     & X3=0     \\
\hline
X1 & X2  &          \\
0  & 0   & 0.891517 \\
   & 1   & 0.310271 \\
1  & 0   & 0.530085 \\
   & 1   & 0.470104 \\
\hline
\end{array}
\quad \quad
\begin{array}{| c | c | l |}
\hline
   &     & R1=0     \\
\hline
X2 & X3  &          \\
0  & 0   & 0.526284 \\
   & 1   & 0.163048 \\
1  & 0   & 0.025951 \\
   & 1   & 0.760911 \\
\hline
\end{array}
\quad \quad
\begin{array}{| c | c | l |}
\hline
   &     & R2=0     \\
\hline
X3 & R1  &          \\
0  & 0   & 0.430232 \\
   & 1   & 0.231407 \\
1  & 0   & 0.713360 \\
   & 1   & 0.741450 \\
\hline
\end{array}
\quad \quad
\begin{array}{| c | c | l |}
\hline
   &     & R3=0     \\
\hline
X1 & X2  &          \\
0  & 0   & 0.529031 \\
   & 1   & 0.404097 \\
1  & 0   & 0.489594 \\
   & 1   & 0.577123 \\
\hline
\end{array}
\end{equation*}

\begin{itemize}
\item With positivity violation
\end{itemize}

\begin{equation*}
\tiny
\begin{array}{r}
    p(X^{(1)}_1=0) = 0.937111
    \\
    \\
    \begin{array}{| c | l |}
    \hline
        & X2=0     \\
    \hline
    X1  &          \\
    0   & 0.231984 \\
    1   & 0.210310 \\
    \hline
    \end{array}
\end{array}
\quad \quad 
\begin{array}{| c | c | l |}
\hline
   &     & X3=0     \\
\hline
X1 & X2  &          \\
0  & 0   & 0.529031 \\
   & 1   & 0.404097 \\
1  & 0   & 0.489594 \\
   & 1   & 0.577123 \\
\hline
\end{array}
\quad \quad 
\begin{array}{| c | c | l |}
\hline
   &     & R1=0     \\
\hline
X2 & X3  &          \\
0  & 0   & 0.492275 \\
   & 1   & 0.476198 \\
1  & 0   & 0.302251 \\
   & 1   & 0.444998 \\
\hline
\end{array}
\quad \quad 
\begin{array}{| c | c | l |}
\hline
   &     & R2=0     \\
\hline
X3 & R1  &          \\
0  & 0   & 0.434144 \\
   & 1   & 0        \\
1  & 0   & 0.698981 \\
   & 1   & 0        \\
\hline
\end{array}
\quad \quad 
\begin{array}{| c | c | l |}
\hline
   &     & R3=0     \\
\hline
X1 & X2  &          \\
0  & 0   & 0.372718 \\
   & 1   & 0.408374 \\
1  & 0   & 0.465549 \\
   & 1   & 0.568108 \\
\hline
\end{array}
\end{equation*}

\newpage

\paragraph{DGPs for graph (d)}

\begin{itemize}
    \item With positivity
\end{itemize}

\begin{equation*}
\tiny
\begin{array}{r}
    p(U_{1} = 0) = 0.412899 \\
    p(U_{2} = 0) = 0.340442 \\
    \begin{array}{| c | l |}
        \hline
            & X^{(1)}_1       \\
        \hline
        U_1  &          \\
        0   & 0.351195 \\
        1   & 0.818716 \\
        \hline
    \end{array}
    \\ \\
    \begin{array}{| c | l |}
        \hline
            & X^{(1)}_2       \\
        \hline
        X^{(1)}_1  &          \\
        0   & 0.589324 \\
        1   & 0.652324 \\
        \hline
    \end{array}
    \\ \\
    \begin{array}{| c | c | l |}
    \hline
       &     & X^{(1)}_3       \\
    \hline
    X^{(1)}_2 & U_1  &          \\
    0  & 0   & 0.296391 \\
       & 1   & 0.484488 \\
    1  & 0   & 0.620866 \\
       & 1   & 0.559600 \\
    \hline
    \end{array}
\end{array}
\quad \quad
\begin{array}{r}
    \begin{array}{| c | c | c | l |}
    \hline
       &    &     & X^{(1)}_4       \\
    \hline
    X^{(1)}_3 & X^{(1)}_2 & X^{(1)}_1  &          \\
    0  & 0  & 0   & 0.357485 \\
       &    & 1   & 0.420410 \\
       & 1  & 0   & 0.235693 \\
       &    & 1   & 0.545312 \\
    1  & 0  & 0   & 0.316772 \\
       &    & 1   & 0.514750 \\
       & 1  & 0   & 0.859519 \\
       &    & 1   & 0.659134 \\
    \hline
    \end{array}
    \\ \\
    \begin{array}{| c | c | c | l |}
    \hline 
       &    &     & R_1       \\
    \hline 
    X^{(1)}_3 & X^{(1)}_4 & U_2  &          \\
    0  & 0  & 0   & 0.547038 \\
       &    & 1   & 0.556173 \\
       & 1  & 0   & 0.091318 \\
       &    & 1   & 0.271319 \\
    1  & 0  & 0   & 0.711250 \\
       &    & 1   & 0.343158 \\
       & 1  & 0   & 0.261107 \\
       &    & 1   & 0.789046 \\
    \hline
    \end{array}
\end{array}
\quad \quad
\begin{array}{| c | c | c | c | l |}
\hline
   &    &    &     & R_2       \\
\hline
X^{(1)}_3 & X^{(1)}_4 & R_1 & U_2  &          \\
0  & 0  & 0  & 0   & 0.634742 \\
   &    &    & 1   & 0.570906 \\
   &    & 1  & 0   & 0.467930 \\
   &    &    & 1   & 0.704905 \\
   & 1  & 0  & 0   & 0.845492 \\
   &    &    & 1   & 0.510379 \\
   &    & 1  & 0   & 0.861527 \\
   &    &    & 1   & 0.829096 \\
1  & 0  & 0  & 0   & 0.465644 \\
   &    &    & 1   & 0.356262 \\
   &    & 1  & 0   & 0.197143 \\
   &    &    & 1   & 0.285206 \\
   & 1  & 0  & 0   & 0.437094 \\
   &    &    & 1   & 0.256127 \\
   &    & 1  & 0   & 0.476503 \\
   &    &    & 1   & 0.590885 \\
\hline
\end{array}
\quad \quad
\begin{array}{r}
    \begin{array}{| c | c | c | l |}
    \hline
              &     &      & R_3      \\
    \hline
    X^{(1)}_1 & R_2 & R_4  &          \\
    0         & 0   & 0    & 0.522999 \\
              &     & 1    & 0.371357 \\
              & 1   & 0    & 0.141733 \\
              &     & 1    & 0.425316 \\
    1         & 0   & 0    & 0.445790 \\
              &     & 1    & 0.538798 \\
              & 1   & 0    & 0.424396 \\
              &     & 1    & 0.462601 \\
    \hline
    \end{array}
    \\ \\
    \begin{array}{| c | l |}
    \hline
        & R_4       \\
    \hline
    X^{(1)}_3  &          \\
    0   & 0.540869 \\
    1   & 0.181146 \\
    \hline
    \end{array}
\end{array}
\end{equation*}

\begin{itemize}
    \item With positivity violation
\end{itemize}

\begin{equation*}
\tiny
\begin{array}{r}
    p(U_{1} = 0) = 0.743035 \\
    p(U_{2} = 0) = 0.390206 \\
    \begin{array}{| c | l |}
        \hline
            & X^{(1)}_1=0 \\
        \hline
        U_1 &          \\
        0   & 0.821949 \\
        1   & 0.760629 \\
        \hline
    \end{array}
    \\ \\
    \begin{array}{| c | l |}
        \hline
                   & X^{(1)}_2=0 \\
        \hline
        X^{(1)}_1  &          \\
        0          & 0.830070 \\
        1          & 0.952262 \\
        \hline
    \end{array}
    \\ \\ 
    \begin{array}{| c | c | l |}
    \hline
              &     & X^{(1)}_3=0 \\
    \hline
    X^{(1)}_2 & U_1 &          \\
    0         & 0   & 0.493445 \\
              & 1   & 0.756955 \\
    1         & 0   & 0.428295 \\
              & 1   & 0.781199 \\
    \hline
    \end{array}
\end{array}
\quad \quad
\begin{array}{r}
    \begin{array}{| c | c | c | l |}
    \hline
              &           &            & X^{(1)}_4=0 \\
    \hline
    X^{(1)}_3 & X^{(1)}_2 & X^{(1)}_1  &          \\
    0         & 0         & 0          & 0.522999 \\
              &           & 1          & 0.371357 \\
              & 1         & 0          & 0.141733 \\
              &           & 1          & 0.425316 \\
    1         & 0         & 0          & 0.445790 \\
              &           & 1          & 0.538798 \\
              & 1         & 0          & 0.424396 \\
              &           & 1          & 0.462601 \\
    \hline
    \end{array}
    \\ \\
    \begin{array}{| c | c | c | l |}
    \hline 
              &           &     & R_1= 0   \\
    \hline 
    X^{(1)}_3 & X^{(1)}_4 & U_2 &          \\
    0         & 0         & 0   & 0.503617 \\
              &           & 1   & 0.365114 \\
              & 1         & 0   & 0.597744 \\
              &           & 1   & 0.541078 \\
    1         & 0         & 0   & 0.822804 \\
              &           & 1   & 0.352675 \\
              & 1         & 0   & 0.464298 \\
              &           & 1   & 0.079114 \\
    \hline
    \end{array}
\end{array}
\quad \quad
\begin{array}{| c | c | c | c | l |}
\hline
          &           &      &     & R_2=0    \\
\hline
X^{(1)}_3 & X^{(1)}_4 & R_1  & U_2 &          \\
0         & 0         & 0    & 0   & 0.394719 \\
          &           &      & 1   & 0.116755 \\
          &           & 1    & 0   & 0.704038 \\
          &           &      & 1   & 0.299178 \\
          & 1         & 0    & 0   & 0.534339 \\
          &           &      & 1   & 0.441878 \\
          &           & 1    & 0   & 0.016644 \\
          &           &      & 1   & 0.482661 \\
1         & 0         & 0    & 0   & 0.639980 \\
          &           &      & 1   & 0.507167 \\
          &           & 1    & 0   & 0.647185 \\
          &           &      & 1   & 0.853005 \\
          & 1         & 0    & 0   & 0.889070 \\
          &           &      & 1   & 0.545897 \\
          &           & 1    & 0   & 0.547497 \\
          &           &      & 1   & 0.511166 \\
\hline
\end{array}
\\ \\
\begin{array}{r}
    \begin{array}{| c | c | c | l |}
    \hline
              &      &      & R_3= 0   \\
    \hline
    X^{(1)}_1 & R_2  & R_4  &          \\
    0         & 0    & 0    & 0.978111 \\
              &      & 1    & 0.294879 \\
              & 1    & 0    & 0.000000 \\
              &      & 1    & 0.590751 \\
    1         & 0    & 0    & 0.345917 \\
              &      & 1    & 0.018613 \\
              & 1    & 0    & 0.000000 \\
              &      & 1    & 0.724987 \\
    \hline
    \end{array}
    \\ \\
    \begin{array}{| c | l |}
    \hline
              &   R_4=0    \\
    \hline
    X^{(1)}_3 &            \\
    0         &   0.250863 \\
    1         &   0.925975 \\
    \hline
    \end{array}
\end{array}
\end{equation*}

\newpage

\paragraph{DGPs for graph (e)}

\begin{itemize}
    \item With positivity violation
\end{itemize}

\begin{equation*}
\tiny
\begin{aligned}
\begin{array}{r}
    p(X^{(1)}_{1}=0) = 0.638138
    \\
    \\
    \begin{array}{| c | l |}
    \hline
       & X2=0     \\
    \hline
    X1 &          \\
    0  & 0.970209 \\
    1  & 0.049576 \\
    \hline
    \end{array}
\end{array}
\quad \quad
\begin{array}{| c | c | l |}
\hline
   &    & X3=0     \\
\hline
X1 & X2 &          \\
0  & 0  & 0.430232 \\
   & 1  & 0.231407 \\
1  & 0  & 0.713360 \\
   & 1  & 0.741450 \\
\hline
\end{array}
\quad \quad
\begin{array}{| c | c | c | l |}
\hline
   &    &    & X4=0     \\
\hline
X1 & X2 & X3 &          \\
0  & 0  & 0  & 0.583809 \\
   &    & 1  & 0.440495 \\
   & 1  & 0  & 0.548947 \\
   &    & 1  & 0.491406 \\
1  & 0  & 0  & 0.413344 \\
   &    & 1  & 0.595662 \\
   & 1  & 0  & 0.452291 \\
   &    & 1  & 0.559238 \\
\hline
\end{array}
\\
\begin{array}{| c | c | c | c | l |}
\hline
   &    &    &     & X5=0     \\
\hline
X1 & X2 & X3 & X4  &          \\
0  & 0  & 0  & 0   & 0.931621 \\
   &    &    & 1   & 0.423657 \\
   &    & 1  & 0   & 0.411324 \\
   &    &    & 1   & 0.656003 \\
   & 1  & 0  & 0   & 0.787735 \\
   &    &    & 1   & 0.091550 \\
   &    & 1  & 0   & 0.988887 \\
   &    &    & 1   & 0.810696 \\
1  & 0  & 0  & 0   & 0.444835 \\
   &    &    & 1   & 0.609901 \\
   &    & 1  & 0   & 0.637284 \\
   &    &    & 1   & 0.422973 \\
   & 1  & 0  & 0   & 0.934693 \\
   &    &    & 1   & 0.069927 \\
   &    & 1  & 0   & 0.527842 \\
   &    &    & 1   & 0.462012 \\
\hline
\end{array}
\quad \quad 
\begin{array}{| c | c | c | c | l |}
\hline
   &    &    &     & R1=0     \\
\hline
X2 & X3 & X4 & X5  &          \\
0  & 0  & 0  & 0   & 0.705047 \\
   &    &    & 1   & 0.513365 \\
   &    & 1  & 0   & 0.340205 \\
   &    &    & 1   & 0.348838 \\
   & 1  & 0  & 0   & 0.613077 \\
   &    &    & 1   & 0.839110 \\
   &    & 1  & 0   & 0.266611 \\
   &    &    & 1   & 0.884113 \\
1  & 0  & 0  & 0   & 0.404455 \\
   &    &    & 1   & 0.158130 \\
   &    & 1  & 0   & 0.442600 \\
   &    &    & 1   & 0.523876 \\
   & 1  & 0  & 0   & 0.431031 \\
   &    &    & 1   & 0.541048 \\
   &    & 1  & 0   & 0.526709 \\
   &    &    & 1   & 0.598381 \\
\hline
\end{array}
\quad \quad
\begin{array}{| c | c | c | c | l |}
\hline
   &    &    &     & R2=0     \\
\hline
X3 & X4 & X5 & R1  &          \\
0  & 0  & 0  & 0   & 0        \\
   &    &    & 1   & 0.456267 \\
   &    & 1  & 0   & 0        \\
   &    &    & 1   & 0.694902 \\
   & 1  & 0  & 0   & 0        \\
   &    &    & 1   & 0.531102 \\
   &    & 1  & 0   & 0        \\
   &    &    & 1   & 0.532139 \\
1  & 0  & 0  & 0   & 0        \\
   &    &    & 1   & 0.229811 \\
   &    & 1  & 0   & 0        \\
   &    &    & 1   & 0.358095 \\
   & 1  & 0  & 0   & 0        \\
   &    &    & 1   & 0.519177 \\
   &    & 1  & 0   & 0        \\
   &    &    & 1   & 0.828434 \\
\hline
\end{array}
\\
\begin{array}{| c | c | c | c | l |}
\hline
   &    &    &     & R3=0     \\
\hline
X1 & X4 & X5 & R2  &          \\
0  & 0  & 0  & 0   & 0.950850 \\
   &    &    & 1   & 0.807773 \\
   &    & 1  & 0   & 0.331640 \\
   &    &    & 1   & 0.940146 \\
   & 1  & 0  & 0   & 0.435004 \\
   &    &    & 1   & 0.033379 \\
   &    & 1  & 0   & 0.959593 \\
   &    &    & 1   & 0.907998 \\
1  & 0  & 0  & 0   & 0.377597 \\
   &    &    & 1   & 0.161617 \\
   &    & 1  & 0   & 0.506291 \\
   &    &    & 1   & 0.479214 \\
   & 1  & 0  & 0   & 0.306019 \\
   &    &    & 1   & 0.745440 \\
   &    & 1  & 0   & 0.944262 \\
   &    &    & 1   & 0.499467 \\
\hline
\end{array}
\quad \quad
\begin{array}{| c | c | c | c | l |}
\hline
   &    &    &     & R4=0     \\
\hline
X2 & X5 & R1 & R3  &          \\
0  & 0  & 0  & 0   & 0.576291 \\
   &    &    & 1   & 0.483073 \\
   &    & 1  & 0   & 0.973214 \\
   &    &    & 1   & 0.618791 \\
   & 1  & 0  & 0   & 0.360446 \\
   &    &    & 1   & 0.633832 \\
   &    & 1  & 0   & 0.084281 \\
   &    &    & 1   & 0.900471 \\
1  & 0  & 0  & 0   & 0.630489 \\
   &    &    & 1   & 0.519740 \\
   &    & 1  & 0   & 0.687443 \\
   &    &    & 1   & 0.545025 \\
   & 1  & 0  & 0   & 0.298542 \\
   &    &    & 1   & 0.209205 \\
   &    & 1  & 0   & 0.124970 \\
   &    &    & 1   & 0.887454 \\
\hline
\end{array}
\quad \quad
\begin{array}{| c | c | c | c | l |}
\hline
   &    &    &     & R5=0     \\
\hline
X1 & X3 & R2 & R4  &          \\
0  & 0  & 0  & 0   & 0.816639 \\
   &    &    & 1   & 0.427999 \\
   &    & 1  & 0   & 0.212742 \\
   &    &    & 1   & 0.277295 \\
   & 1  & 0  & 0   & 0.516575 \\
   &    &    & 1   & 0.263809 \\
   &    & 1  & 0   & 0.776928 \\
   &    &    & 1   & 0.930711 \\
1  & 0  & 0  & 0   & 0.816094 \\
   &    &    & 1   & 0.577293 \\
   &    & 1  & 0   & 0.908045 \\
   &    &    & 1   & 0.823171 \\
   & 1  & 0  & 0   & 0.240437 \\
   &    &    & 1   & 0.678328 \\
   &    & 1  & 0   & 0.417560 \\
   &    &    & 1   & 0.742792 \\
\hline
\end{array}
\end{aligned}
\end{equation*}

\newpage

\paragraph{DGPs for graph (f)}

We randomly choose a DGP as follow

\begin{enumerate}
    \item For each $i=1, \ldots, 20$, sample random numbers in $(0,1)$ and assign to $p(X^{(1)}_i = 0 \mid \pa_{\mathcal{G}}(X^{(1)}_i))$.
    \item For each $i=1, \ldots, 20$, we assign values to $p(R_i = 0 \mid \pa_{\mathcal{G}}(R_i))$ as follow. Note that $R_{i+k} \in \pa_{\mathcal{G}}(R_i)$ for all $i = 1, \ldots, 20$ and $k = 0, \ldots, 10$ such that $i+k \leq 20$.
    \begin{enumerate}
        \item If $R_{i+k}=0$ for any $k \geq 3$, then $p(R_i = 0 \mid \pa_{\mathcal{G}}(R_i)) = 0$.
        \item If $R_{i+1}=0$ and/or $R_{i+2}=0$, then sample random numbers in $(0, 0.5)$, then multiply by $1.8$, and assign to the corresponding $p(R_i = 0 \mid \pa_{\mathcal{G}}(R_i))$. The goal is to create values higher than 0.5 when $R_{i+1}=0$ and/or $R_{i+2}=0$.
        \item Otherwise, sample random numbers in $(0, 1)$ and assign to $p(R_i = 0 \mid \pa_{\mathcal{G}}(R_i))$.
    \end{enumerate}
\end{enumerate}

\end{document}